\documentclass[british,english]{article}
\usepackage{lmodern}

\usepackage[T1]{fontenc}
\usepackage[latin9]{inputenc}
\usepackage{color}
\usepackage{babel}
\usepackage{array}
\usepackage{prettyref}
\usepackage{wrapfig}
\usepackage{multirow}
\usepackage{amsthm}
\usepackage{amsmath}
\usepackage{amssymb}
\usepackage{graphicx}
\usepackage[numbers]{natbib}
\usepackage{subscript}
\usepackage[unicode=true,pdfusetitle,
 bookmarks=true,bookmarksnumbered=false,bookmarksopen=false,
 breaklinks=false,pdfborder={0 0 0},backref=page,colorlinks=true]
 {hyperref}
\hypersetup{
 linkcolor=blue, urlcolor=blue, citecolor=blue,  pdfpagemode=UseNone}

\makeatletter

\newcommand{\noun}[1]{\textsc{#1}}
\providecommand{\tabularnewline}{\\}

\usepackage{enumitem}		
\newcommand{\lyxaddress}[1]{
\par {\raggedright #1
\vspace{1.4em}
\noindent\par}
}
\theoremstyle{plain}
\newtheorem{thm}{\protect\theoremname}
  \theoremstyle{definition}
  \newtheorem{defn}[thm]{\protect\definitionname}
  \theoremstyle{plain}
  \newtheorem{lem}[thm]{\protect\lemmaname}
  \theoremstyle{definition}
  \newtheorem{example}[thm]{\protect\examplename}
  \theoremstyle{remark}
  \newtheorem{rem}[thm]{\protect\remarkname}
  \theoremstyle{plain}
  \newtheorem{cor}[thm]{\protect\corollaryname}

\@ifundefined{date}{}{\date{}}

\newrefformat{pro}{Proposition~\ref{#1}}
\newrefformat{cor}{Corollary~\ref{#1}}
\newrefformat{def}{Definition~\ref{#1}}
\newrefformat{sub}{Section~\ref{#1}}
\newrefformat{exa}{Example~\ref{#1}}
\newrefformat{rem}{Remark~\ref{#1}}

\usepackage{url}

\setlist{noitemsep}
\setlist[1]{leftmargin=1.7em}
\setlist[2]{leftmargin=1.7em}

\usepackage[font=small, width=.9\textwidth]{caption}

\newlength{\LyXMinipageIndent}
\usepackage{titlesec}
\titleformat{\section}{\large\bfseries}{\thesection}{1em}{}
\titleformat{\subsection}{\normalsize\bfseries}{\thesubsection}{1em}{}

\makeatother

  \addto\captionsbritish{\renewcommand{\corollaryname}{Corollary}}
  \addto\captionsbritish{\renewcommand{\definitionname}{Definition}}
  \addto\captionsbritish{\renewcommand{\examplename}{Example}}
  \addto\captionsbritish{\renewcommand{\lemmaname}{Lemma}}
  \addto\captionsbritish{\renewcommand{\remarkname}{Remark}}
  \addto\captionsbritish{\renewcommand{\theoremname}{Theorem}}
  \addto\captionsenglish{\renewcommand{\corollaryname}{Corollary}}
  \addto\captionsenglish{\renewcommand{\definitionname}{Definition}}
  \addto\captionsenglish{\renewcommand{\examplename}{Example}}
  \addto\captionsenglish{\renewcommand{\lemmaname}{Lemma}}
  \addto\captionsenglish{\renewcommand{\remarkname}{Remark}}
  \addto\captionsenglish{\renewcommand{\theoremname}{Theorem}}
  \providecommand{\corollaryname}{Corollary}
  \providecommand{\definitionname}{Definition}
  \providecommand{\examplename}{Example}
  \providecommand{\lemmaname}{Lemma}
  \providecommand{\remarkname}{Remark}
\providecommand{\theoremname}{Theorem}

\begin{document}

\title{\vspace{-6.5ex}A Computational Trichotomy\\
for Connectivity of Boolean Satisfiability}

\author{\noun{\normalsize{}Konrad W. Schwerdtfeger}}

\maketitle

\lyxaddress{\begin{center}
{\small{}\vspace{-5ex}Institut für Theoretische Informatik, Leibniz
Universität Hannover,}\\
{\small{}Appelstr. 4, 30167 Hannover, Germany}\\
{\small{}\path|k.w.s@gmx.net|}
\par\end{center}}
\begin{abstract}
For Boolean satisfiability problems, the structure of the solution
space is characterized by the solution graph, where the vertices are
the solutions, and two solutions are connected iff they differ in
exactly one variable. In 2006, Gopalan et al.\ studied connectivity
properties of the solution graph and related complexity issues for
CSPs \citep{gop}, motivated mainly by research on satisfiability
algorithms and the satisfiability threshold. They proved dichotomies
for the diameter of connected components and for the complexity of
the $st$-connectivity question, and conjectured a trichotomy for
the connectivity question.

Building on this work, we here prove the trichotomy: Connectivity
is either in P, coNP-complete, or PSPACE-complete. Also, we correct
a minor mistake in \citep{gop}, which leads to a slight shift of
the boundaries towards the hard side.\\

\end{abstract}

\section{Introduction}

In 2006, P.~Gopalan, P.~G.~Kolaitis, E.~Maneva, and C.~H.~Papadimitriou
investigated connectivity properties of the solution space of Boolean
constraint satisfaction problems \citep{Gopalan:2006,gop}. Their
work was motivated inter alia by research on heuristics for satisfiability
algorithms and threshold phenomena. Indeed, the solution space connectivity
is strongly correlated to the performance of standard satisfiability
algorithms like WalkSAT and DPLL on random instances: As one approaches
the\emph{ satisfiability threshold} (the ratio of constraints to variables
at which random $k$-CNF-formulas become unsatisfiable for $k\geq3$)
from below, the solution space fractures, and the performance of the
algorithms breaks down \citep{mezard2005clustering,maneva2007new}.
These insights mainly came from statistical physics, and lead to the
development of the \emph{survey propagation algorithm}, which has
much better performance on random instances \citep{maneva2007new}.

Meanwhile, Gopalan et al.'s results have also been applied directly
to \emph{reconfiguration problems}, that arise when a step-by-step
transformation between two feasible solutions of a problem is searched,
such that all intermediate results are also feasible. Recently, the
reconfiguration versions of many problems such as \noun{Independent-Set},
\noun{Vertex-Cover}, \noun{Set-Cover,} \noun{Graph-$k$-Coloring},
\noun{Shor\-test-Path} have been studied \citep{recon,short}, and
many complexity results were obtained. Another related problem for
which the solution space connectivity could be of interest is \emph{structure
identification}, where one is given a relation explicitly and seeks
a short representation of some kind \citep{creignou2008structure};
this problem is important especially in artificial intelligence.

The solutions (satisfying assignments) of a formula $\phi$ over $n$
variables induce a subgraph $G(\phi)$ of the $n$-dimensional hypercube
graph, that is, the vertices are the solutions of $\phi$, and two
solutions are connected iff they differ in exactly one variable.

\begin{figure}[!h]
\begin{centering}
\includegraphics[scale=0.64]{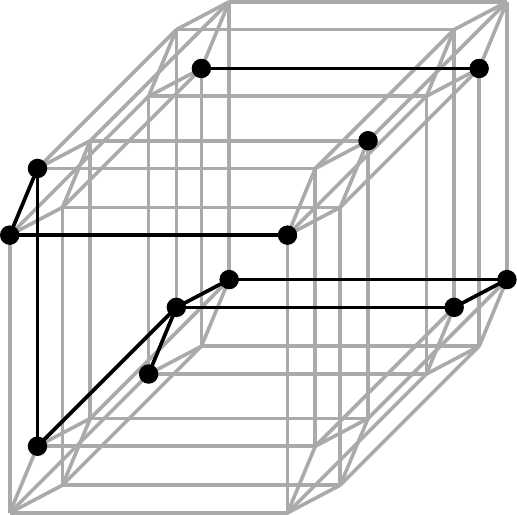}$\qquad$\includegraphics[scale=0.34]{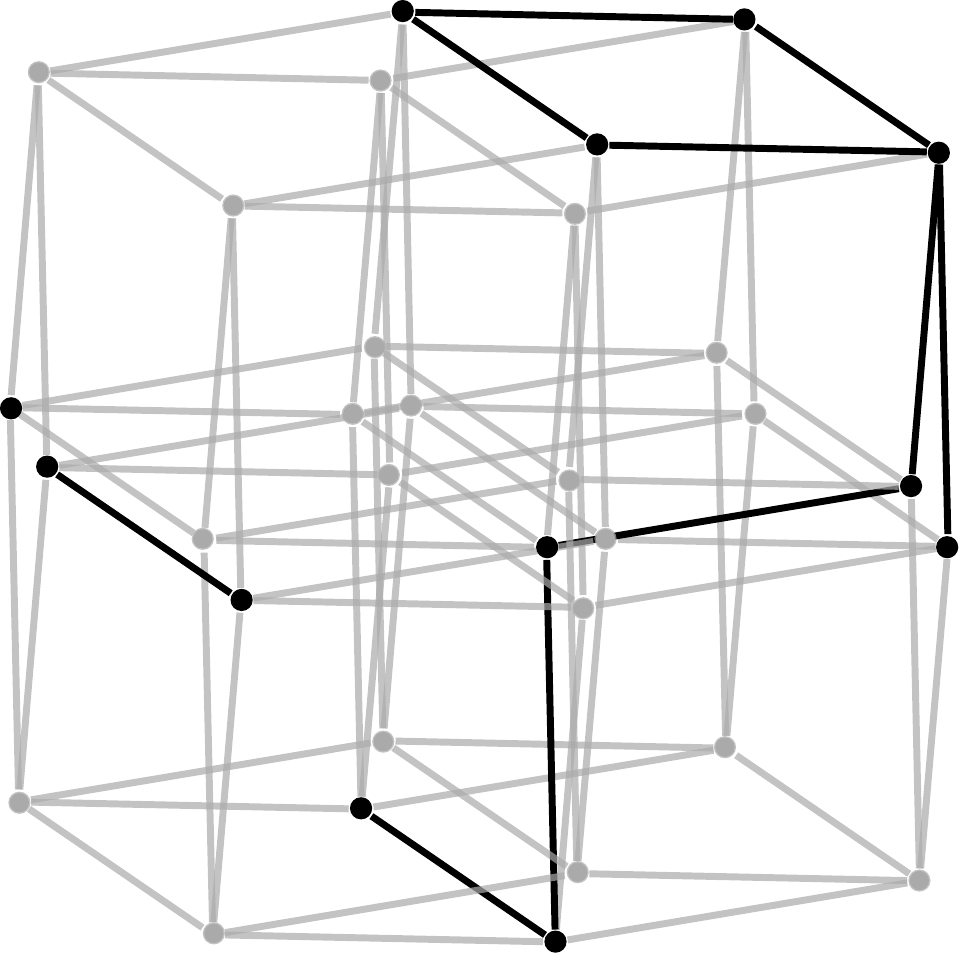}$\qquad$\includegraphics[scale=0.36]{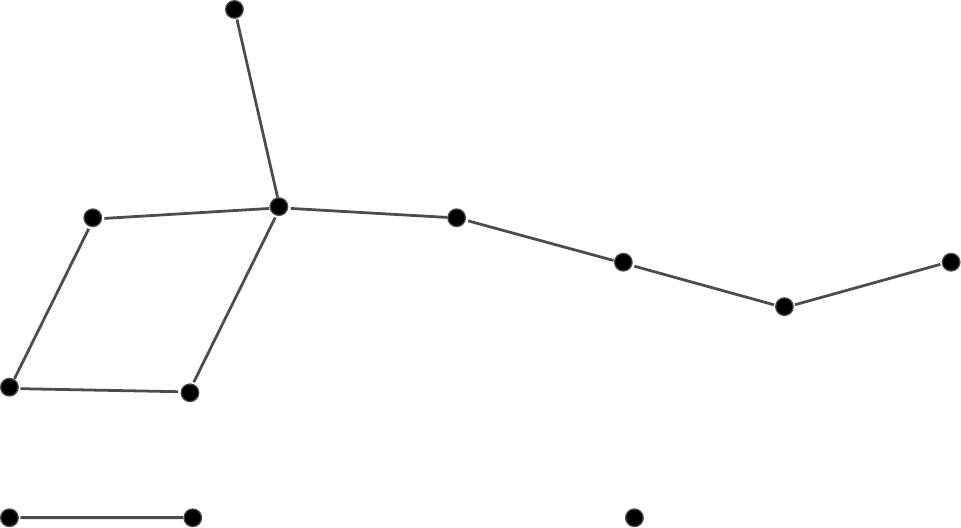}
\par\end{centering}

\protect\caption{Depictions of the subgraph of the 5-dimensional hypercube graph induced
by a typical random Boolean relation with 12 elements. Left: highlighted
on a orthographic hypercube projection. Center: highlighted on a ``Spectral
Embedding'' of the hypercube graph by \noun{Mathematica}.}
\end{figure}

\begin{figure}[!h]
\begin{centering}
\includegraphics[scale=0.3]{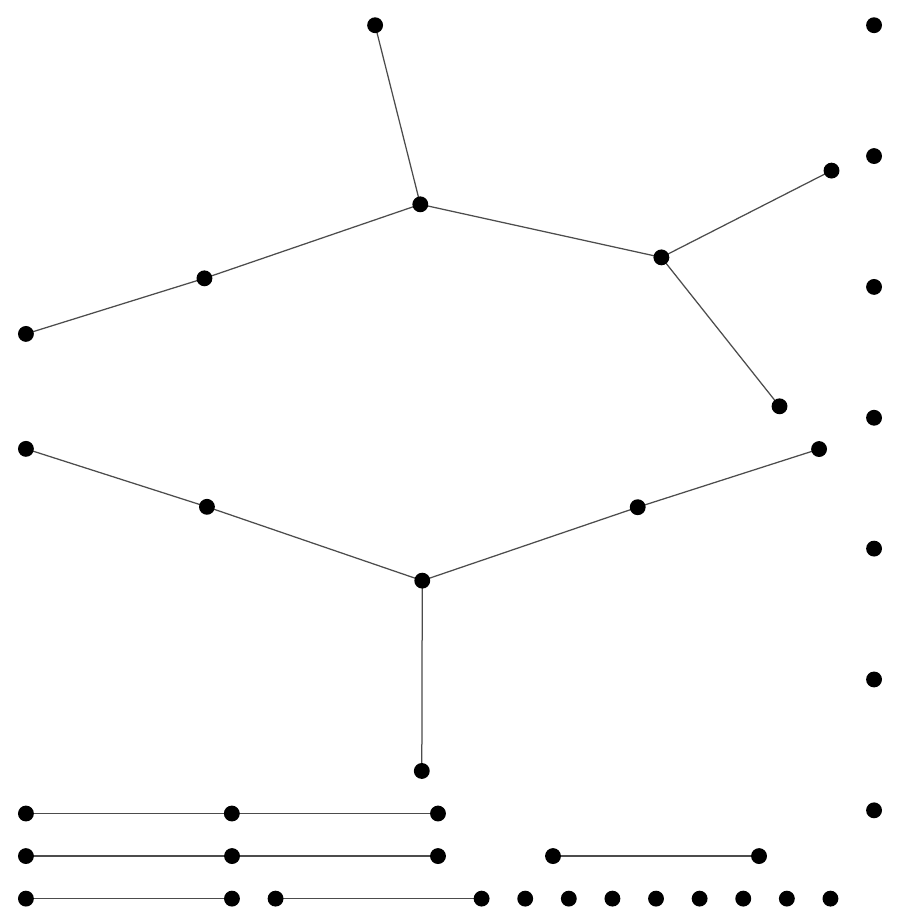}$\qquad$\includegraphics[scale=0.36]{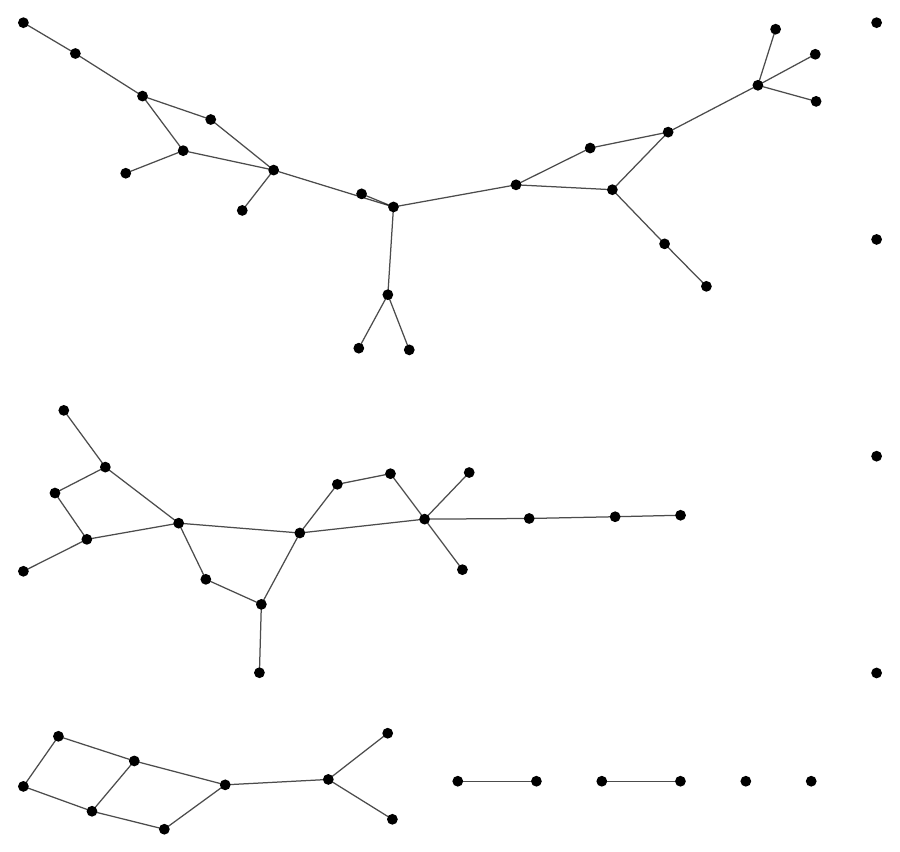}$\qquad$\includegraphics[scale=0.46]{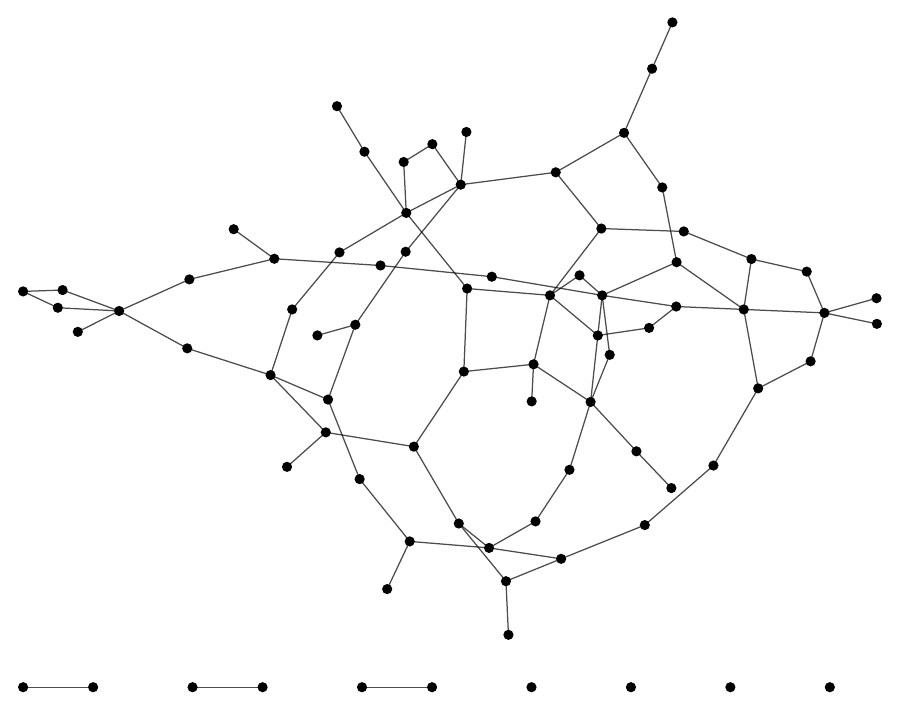}
\par\end{centering}

\protect\caption{Subgraphs of the 8-dimensional hypercube graph (with 256 vertices)
induced by typical random relations with 40, 60 and 80 elements.}

\end{figure}

Gopalan et al.\ specifically addressed CNF\textsubscript{C}($\mathcal{S}$)-formulas
(CNF($\mathcal{S}$)-formulas with constants), see \prettyref{def:rel},
and studied the complexity of the following two decision problems,
\begin{itemize}
\item the connectivity problem \noun{Conn}\textsubscript{C}($\mathcal{S}$),
that asks for a given CNF\textsubscript{C}($\mathcal{S}$)-formula
$\phi$ whether $G(\phi)$ is connected,
\item the $st$-connectivity problem \noun{st-Conn}\textsubscript{C}($\mathcal{S}$),
that asks for a given CNF\textsubscript{C}($\mathcal{S}$)-formula
$\phi$ and two solutions $\boldsymbol{s}$ and $\boldsymbol{t}$
whether there a path from $\boldsymbol{s}$ to $\boldsymbol{t}$ in
$G(\phi)$.
\end{itemize}
Also, they considered
\begin{itemize}
\item the maximal diameter of any connected component of $G(\phi)$ for
a CNF\textsubscript{C}($\mathcal{S}$)-formula $\phi$, where the
diameter of a component is the maximal shortest-path distance between
any two vectors in that component.
\end{itemize}
They established a common structural and computational dichotomy,
and introduced the corresponding class of \emph{tight} sets of relations,
which properly contains all Schaefer sets of relations, see \prettyref{def:cpss}:
For tight sets $\mathcal{S}$, the diameter is linear in the number
of variables, \noun{st-Conn}\textsubscript{C}($\mathcal{S}$) is
in P and \noun{Conn}\textsubscript{C}($\mathcal{S}$) is in coNP,
while on the other side, the diameter can be exponential, and both
problems are PSPACE-complete. Their results are summarized in comparison
to the satisfiability problem \emph{\noun{Sat}}\textsubscript{C}(\emph{$\mathcal{S}$})
in the table below.

\begin{table}[!h]
\noindent \begin{centering}
\renewcommand{\arraystretch}{1.1}
\begingroup\tabcolsep=3pt%
\begin{tabular}{|c||c|c|c|c|}
\hline 
{\footnotesize{}$\mathcal{S}$} & \noun{\footnotesize{}Sat}{\footnotesize{}\textsubscript{C}($\mathcal{S}$) } & \emph{\noun{\footnotesize{}st-}}\noun{\footnotesize{}Conn}{\footnotesize{}\textsubscript{C}($\mathcal{S}$)} & \noun{\footnotesize{}Conn}{\footnotesize{}\textsubscript{C}($\mathcal{S}$)} & {\footnotesize{}Diameter}\tabularnewline
\hline 
\hline 
{\footnotesize{}Schaefer} & {\footnotesize{}P} & \multirow{2}{*}{{\footnotesize{}P}} & {\footnotesize{}coNP} & \multirow{2}{*}{{\footnotesize{}$O(n)$}}\tabularnewline
\cline{2-2} \cline{4-4} 
{\footnotesize{}Tight, not Schaefer} & \multirow{2}{*}{{\footnotesize{}NP-compl.}} &  & {\footnotesize{}coNP-compl.} & \tabularnewline
\cline{3-5} 
{\footnotesize{}Not tight} &  & {\footnotesize{}PSPACE-compl.} & {\footnotesize{}PSPACE-compl.} & {\footnotesize{}$2^{\Omega(\sqrt{n})}$}\tabularnewline
\hline 
\end{tabular}\endgroup
\par\end{centering}

\protect\caption{Gopalan et al.'s results \citep{gop}}

\end{table}

Moreover, they conjectured a trichotomy for \noun{Conn}\textsubscript{C}($\mathcal{S}$):
For a certain sub-class of Schaefer sets of relations, \noun{Conn}\textsubscript{C}($\mathcal{S}$)
is in P, while for all other tight sets it is coNP-complete.

In \prettyref{sec:imp} we will argue that Gopalan et al.~did not
consider repeated occurrences of variables in constraint applications.
As we will see there, repeated variables can make the problems harder
and the diameter exponential in some cases, which leads to a slight
shift of the boundaries. Also, corrections are needed in the application
of Gopalan et al.'s concept of structural expressibility, used in
the reductions from 3-CNF formulas that establish PSPACE-completeness.

In \prettyref{sec:tri}, we prove the conjectured trichotomy for \noun{Conn}\textsubscript{C}($\mathcal{S}$),
also with the boundaries shifted in the hard direction. Fitted to
the correct boundaries, we will introduce the classes of \emph{safely
tight} and \emph{CPSS} sets of relations; The supplemental \prettyref{sec:cpss}
will investigate certain properties of CPSS sets of relations. The
following table summarizes our results.\\
\begin{table}[!h]
\noindent \begin{centering}
\renewcommand{\arraystretch}{1.1}
\begingroup\tabcolsep=3pt%
\begin{tabular}{|c||c|c|c|c|}
\hline 
{\footnotesize{}$\mathcal{S}$} & \noun{\footnotesize{}Sat}{\footnotesize{}\textsubscript{C}($\mathcal{S}$) } & \emph{\noun{\footnotesize{}st-}}\noun{\footnotesize{}Conn}{\footnotesize{}\textsubscript{C}($\mathcal{S}$)} & \noun{\footnotesize{}Conn}{\footnotesize{}\textsubscript{C}($\mathcal{S}$)} & {\footnotesize{}Diameter}\tabularnewline
\hline 
\hline 
{\footnotesize{}CPSS} & \multirow{2}{*}{{\footnotesize{}P}} & \multirow{3}{*}{{\footnotesize{}P}} & {\footnotesize{}P} & \multirow{3}{*}{{\footnotesize{}$O(n)$}}\tabularnewline
\cline{4-4} 
{\footnotesize{}Schaefer, not CPSS} &  &  & \multirow{2}{*}{{\footnotesize{}coNP-compl.}} & \tabularnewline
\cline{2-2} 
{\footnotesize{}Safely tight, not Schaefer} & \multirow{2}{*}{{\footnotesize{}NP-compl.}} &  &  & \tabularnewline
\cline{3-5} 
{\footnotesize{}Not safely tight} &  & {\footnotesize{}PSPACE-compl.} & {\footnotesize{}PSPACE-compl.} & {\footnotesize{}$2^{\Omega(\sqrt{n})}$}\tabularnewline
\hline 
\end{tabular}\endgroup
\par\end{centering}

\protect\caption{Complete classification of the connectivity problems and the diameter
for CNF($\mathcal{S}$)-formulas with constants, in comparison to
\noun{Sat.}}

\end{table}

\section{Preliminaries}

First we introduce some terminology for Boolean relations and formulas.
We will use the standard notions also used in \citep{gop}, but carefully
define \emph{substitution of constants} and \emph{identification of
variables}, and distinguish CNF($\mathcal{S}$)-formulas with and
without constants.
\begin{defn}
\label{def:rel}An \emph{$n$-ary Boolean relation} (or\emph{ logical
relation, relation} for short) is a subset of $\{0,1\}^{n}$ ($n\ge1$).

For an $n$-ary relation $R$, we can define an $(n-k)$-ary relation
\[
R'(x_{1},\ldots,x_{n-k})=R(\xi_{1},\ldots,\xi_{n})
\]
($0<k<n$). If each $\xi_{i}\in\{0,1,x_{1},\ldots,x_{n-k}\}$ and
each variable $x_{i}\in\{x_{1},\ldots,x_{n-k}\}$ occurs at most once
in $(\xi_{1},\ldots,\xi_{n})$, we say $R'$ is obtained from $R$
by \emph{substitution of constants}. If each $\xi_{i}\in\{x_{1},\ldots,x_{n-k}\}$,
(and each $x_{i}\in\{x_{1},\ldots,x_{n-k}\}$ may occur any number
of times in $(\xi_{1},\ldots,\xi_{n})$), $R'$ is obtained by \emph{identification
of variables}. Note that we allow any permutation of the variables
in both cases.

The set of solutions of a propositional formula $\phi$ over $n$
variables defines in a natural way an $n$-ary relation $[\phi]$,
where the variables are taken in lexicographic order. We will often
identify the formula $\phi$ with the relation it defines and omit
the brackets.
\end{defn}
In the following definition note that we write \noun{st-Conn}\textsubscript{C}($\mathcal{S}$)
resp.~\noun{Conn}\textsubscript{C}($\mathcal{S}$) instead of \noun{st-Conn}($\mathcal{S}$)
resp.\ \noun{st-Conn}($\mathcal{S}$) like Gopalan et al., for consistency
with the usual notation \noun{Sat}($\mathcal{S}$) for the satisfiability
problem without constants and \emph{\noun{S}}\noun{at}\textsubscript{C}($\mathcal{S}$)
for the one with constants. Accordingly, we call CNF($\mathcal{S}$)-formulas
with constants CNF\textsubscript{C}($\mathcal{S}$)-formulas.
\begin{defn}
\label{def:cnf}A \emph{CNF-formula} is a propositional formula of
the form $C_{1}\wedge\cdots\wedge C_{m}$ ($1\leq m<\infty$), where
each $C_{i}$ is a \emph{clause}, that is, a finite disjunction of
\emph{literals} (variables or negated variables). A\emph{ $k$-CNF-formula}
($k\geq1$) is a CNF-formula where each $C_{i}$ has at most $k$
literals. A \emph{Horn (dual Horn)} formula is a CNF-formula where
each $C_{i}$ has at most one positive (negative) literal.

For a finite set of relations $\mathcal{S}$, a\emph{ CNF\textsubscript{C}($\mathcal{S}$)-formula}
over a set of variables $V$ is a finite conjunction $C_{1}\wedge\cdots\wedge C_{m}$,
where each $C_{i}$ is a \emph{constraint application} (\emph{constraint
}for short), i.e., an expression of the form $R(\xi_{1},\ldots,\xi_{k})$,
with a $k$-ary relation $R\in\mathcal{S}$, and each $\xi_{j}$ is
a variable from $V$ or one of the constants 0, 1. By $\mathrm{Var}(C_{i})$,
we denote the set of variables occurring in $\xi_{1},\ldots,\xi_{k}$.
With the\emph{ relation corresponding to $C_{i}$} we mean the relation
$[R(\xi_{1},\ldots,\xi_{k})]$ (that may be different from $R$ by
substitution of constants, or identification or permutation of variables).
A\emph{ CNF($\mathcal{S}$)-formula }is a\emph{ }CNF\textsubscript{C}($\mathcal{S}$)-formula
where each $\xi_{j}$ is a variable in $V$, not a constant.
\end{defn}
We define the solution graph and its diameter as in \citep{gop}.
We use $\boldsymbol{a},\boldsymbol{b},\ldots$ or $\boldsymbol{a}^{1},\boldsymbol{a}^{2},\ldots$
to denote vectors of Boolean values and $\boldsymbol{x},\boldsymbol{y},\ldots$
or $\boldsymbol{x}^{1},\boldsymbol{x}^{2},\ldots$ to denote vectors
of variables, $\boldsymbol{a}=(a_{1},a_{2},\ldots)$ and $\boldsymbol{x}=(x_{1},x_{2},\ldots)$.
\begin{defn}
The \emph{solution graph} $G(\phi)$ of $\phi$ is the subgraph of
the $n$-dimensional hypercube graph induced by the vectors in $[\phi]$,
i.e., the vertices of $G(\phi$) are the vectors in $[\phi]$, and
there is an edge between two vectors iff they differ in exactly one
variable. We will also refer to $G(R)$ for any logical relation $R$
(not necessarily defined by a formula).

If \textbf{$\boldsymbol{a}$} and $\boldsymbol{b}$ are solutions
of a formula $\phi$ and lie in the same connected component (\emph{component}
for short) of $G(\phi)$, we write $d_{\phi}(\boldsymbol{a},\boldsymbol{b})$
to denote the shortest-path distance between \textbf{$\boldsymbol{a}$}
and $\boldsymbol{b}$. The \emph{diameter of a component} is the maximal
shortest-path distance between any two vectors in that component.
The\emph{ diameter of $G(\phi)$} is the maximal diameter of any component.

The \emph{Hamming distance }$|\boldsymbol{a}-\boldsymbol{b}|$ of
two Boolean vectors \textbf{$\boldsymbol{a}$} and $\boldsymbol{b}$
is the number of positions in which they differ.
\end{defn}
We define the following decision problems for CNF($\mathcal{S}$)-formulas
resp. CNF\textsubscript{C}($\mathcal{S}$)-formulas:
\begin{itemize}
\item the \emph{satisfiability problem }\emph{\noun{Sat}}\emph{($\mathcal{S}$):}
Given a CNF($\mathcal{S}$)-formula $\phi$, is $\phi$ satisfiable?
\item the \emph{satisfiability problem with constants }\emph{\noun{S}}\noun{at}\textsubscript{\emph{C}}(\emph{$\mathcal{S}$):}
Given a CNF\textsubscript{C}($\mathcal{S}$)-formula $\phi$, is
$\phi$ satisfiable?
\item the\emph{ connectivity problem (with constants) }\noun{Conn}\textsubscript{\emph{C}}($\mathcal{S}$)\emph{:}
Given a CNF\textsubscript{C}($\mathcal{S}$)-formula $\phi$, is
$G(\phi)$ connected? (if $\phi$ is unsatisfiable, then $G(\phi)$
is considered connected)
\item the\emph{ $st$-connectivity problem (with constants) }\emph{\noun{st-}}\noun{Conn}\textsubscript{\emph{C}}($\mathcal{S}$)\emph{\noun{:}}\noun{
}Given a CNF\textsubscript{C}($\mathcal{S}$)-formula $\phi$ and
two solutions $\boldsymbol{s}$ and $\boldsymbol{t}$, is there a
path from $\boldsymbol{s}$ to $\boldsymbol{t}$ in $G(\phi)$?
\end{itemize}
The complexity of the problems depends on the kind of relations in
$\mathcal{S}$; we now define the relevant types. Some are already
familiar from Schaefer's classification of \emph{\noun{Sat,}} some
were introduced by Gopalan et al., and the ones starting with ``safely''
are new; \emph{IHSB }stands for\emph{ }``implicative hitting set-bounded''
and was adopted by Gopalan et al. from \citep{Creignou:2001:CCB:377810},
where it was introduced for a refinement of Schaefer\textquoteright s
theorem and the classification of related problems.
\begin{defn}
\label{def:typ}Let $R$ be an $n$-ary logical relation.
\begin{itemize}
\item $R$ is\emph{ 0-valid (1-valid)} if $0^{n}\in R$ ($1^{n}\in R$).
\item $R$ is \emph{bijunctive} if it is the set of solutions of a 2-CNF-formula.
\item $R$ is \emph{Horn (dual Horn)} if it is the set of solutions of a
Horn (dual Horn) formula.
\item $R$ is \emph{affine} if it is the set of solutions of a formula $x_{i_{1}}\oplus\ldots\oplus x_{i_{m}}\oplus c$
with $i_{1},\ldots,i_{m}\in\{1,\ldots,n\}$ and $c\in\{0,1\}$.
\item $R$ is \emph{componentwise bijunctive} if every connected component
of $G(R)$ is a bijunctive relation. $R$ is \emph{safely componentwise
bijunctive} if $R$ and every relation $R'$ obtained from $R$ by
identification of variables is componentwise bijunctive.
\item $R$ is \emph{OR-free} (\emph{NAND-free}) if the relation OR = $\left\{ 01,10,11\right\} $
(NAND = $\left\{ 00,01,10\right\} $) cannot be obtained from $R$
by substitution of constants. $R$ is \emph{safely OR-free} (\emph{safely
NAND-free}) if $R$ and every relation $R'$ obtained from $R$ by
identification of variables is OR-free (NAND-free).
\item $R$ is \emph{IHSB$-$ (IHSB$+$)} if it is the set of solutions of
a Horn (dual Horn) formula in which all clauses with more than 2 literals
have only negative literals (only positive literals).
\item $R$ is \emph{componentwise IHSB$-$ (componentwise IHSB$+$)} if
every connected component of $G(R)$ is\emph{ }IHSB$-$ (IHSB$+$).
$R$ is \emph{safely} \emph{componentwise IHSB$-$ (safely componentwise
IHSB$+$) }if $R$ and every relation $R'$ obtained from $R$ by
identification of variables is componentwise\emph{ }IHSB$-$ (componentwise
IHSB$+$).
\end{itemize}
\end{defn}
If one is given the relation explicitly (as a set of vectors), the
properties 0-valid, 1-valid, OR-free\emph{ }and\emph{ }NAND-free can
be checked easily. Bijunctive, Horn, dual Horn, affine, IHSB$-$ and
IHSB$+$ can be checked by \emph{closure properties}:
\begin{defn}
A relation $R$ is \emph{closed} under some $n$-ary operation $f$
iff the vector obtained by the coordinate-wise application of $f$
to any $m$ vectors from $R$ is again in $R$, i.e., if 
\[
\boldsymbol{a}^{1},\ldots,\boldsymbol{a}^{m}\in R\Longrightarrow(f(a_{1}^{1},\ldots,a_{1}^{m}),\ldots f(a_{n}^{1},\ldots,a_{n}^{m}))\in R.
\]
\end{defn}
\begin{lem}
\label{lem:cl}A relation $R$ is
\begin{itemize}
\item bijunctive, iff it is closed under the ternary majority operation\\
 \emph{MAJ($x,y,z$)}=$\left(x\vee y\right)\wedge\left(y\vee z\right)\wedge\left(z\vee x\right)$
\citep[Lemma 4.9]{Creignou:2001:CCB:377810},
\item Horn (dual Horn), iff it is closed under $\wedge$ (under $\vee$,
resp.) \citep[Lemma 4.8]{Creignou:2001:CCB:377810},
\item affine,\emph{ }iff it is closed under $x\oplus y\oplus z$ \citep[Lemma 4.10]{Creignou:2001:CCB:377810},
\item IHSB$-$ (IHSB$+$), iff it is closed under $x\wedge(y\vee z)$ (under
$x\vee(y\wedge z)$, resp.).
\end{itemize}
\end{lem}
\begin{proof}
For IHSB$-$ and IHSB$+$ relations, this can be verified using the
Galois correspondence between closed sets of relations and closed
sets of Boolean functions (see \citep{bohler2005bases}). From the
table in \citep{bohler2005bases} we find that the IHSB$-$ relations
are a base of the co-clone INV($\mathsf{S}_{10}$), and the IHSB$+$
ones a base of INV($\mathsf{S}_{00}$), and from the table in \citep{bloc}
we see that $x\wedge(y\vee z)$ and $x\vee(y\wedge z)$ are bases
of the clones $\mathsf{S}_{10}$ and $\mathsf{S}_{00}$, resp.
\end{proof}
The following examples show that the ``safely'' classes are properly
contained in the corresponding ``unsafe'' ones.
\begin{example}
\label{exa:or}The relation $\{001,110,111\}$ is OR-free, but not
safely OR-free, as identifying the first two variables gives $\{01,10,11\}$.
\end{example}
The smallest examples of relations that are componentwise bijunctive,
but not safely componentwise bijunctive, or Horn and componentwise
\emph{IHSB$-$}, but not safely componentwise \emph{IHSB$-$} are
of dimension 4:
\begin{example}
\label{exa:ecnp}For the relation $R_{\mathrm{coNP}}=\{0000,0100,1100,\,\,0011,1011\}$,
both components $\{0000,0100,1100\}$ and $\{0011,1011\}$ are closed
under MAJ and under $x\vee(y\wedge z)$, but the relation $R'=\{000,010,110,001,101\}$,
obtained from $R_{\mathrm{coNP}}$ by identifying the third and fourth
variable, has only one component that is neither closed under MAJ,
nor under $x\wedge(y\vee z)$: applying MAJ or $x\wedge(y\vee z)$
coordinate-wise to ($110,000,101$) both gives $100\notin R'$.

For an example of a formula consider 
\[
\phi_{\mathrm{coNP}}=(x\wedge y)\vee\left(\overline{x}\wedge\overline{y}\wedge(\overline{z}\vee\overline{w})\right)\equiv\left(x\vee\overline{y}\right)\wedge\left(\overline{x}\vee y\right)\wedge\left(x\vee\overline{z}\vee\overline{w}\right)\wedge\left(y\vee\overline{z}\vee\overline{w}\right),
\]
which is clearly componentwise bijunctive and componentwise IHSB$-$,
but $x\vee\overline{z}\vee\overline{w}$, obtained by identifying
$y$ with $x$, has only one component that is neither bijunctive
nor IHSB$-$.

\begin{figure}[!h]
\begin{centering}
\includegraphics[scale=0.7]{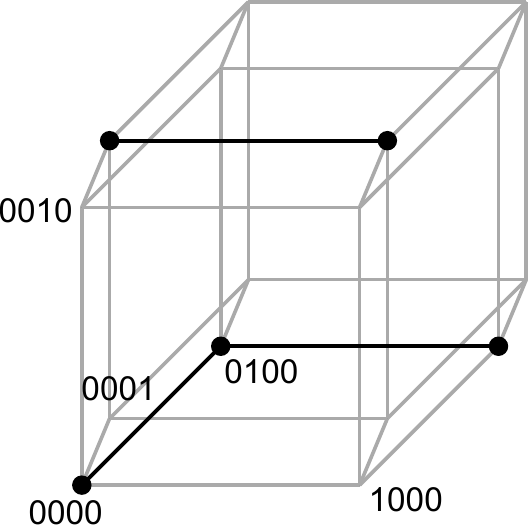}$\qquad$\includegraphics[scale=0.7]{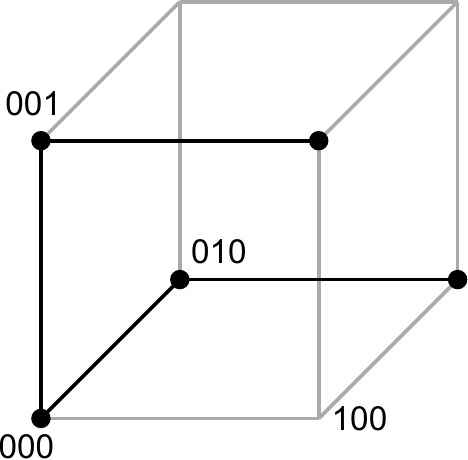}
\par\end{centering}

\protect\caption{The solution graphs of the relations $R_{\mathrm{coNP}}$ and $R'$
from Example \ref{exa:ecnp}, drawn on orthographic hypercube projections;
the ``axis vertices'' are labeled.}
\end{figure}

\end{example}
The following classes of sets of relations are fitted to the structural
and computational boundaries for the connectivity; the term \emph{CPSS}
stands for \emph{constraint-projection separating Schaefer} and will
become clear in \prettyref{sec:cpss} from Definition \ref{dcp} and
\prettyref{lem:prjs b-1}.
\begin{defn}
\label{def:cpss}A set $\mathcal{S}$ of logical relations is \emph{tight}
(\emph{safely tight}) if at least one of the following conditions
holds:
\begin{enumerate}
\item every relation in $\mathcal{S}$ is componentwise bijunctive (safely
componentwise bijunctive).
\item every relation in $\mathcal{S}$ is OR-free (safely OR-free).
\item every relation in $\mathcal{S}$ is NAND-free (safely NAND-free).
\end{enumerate}
A set $\mathcal{S}$ of logical relations is \emph{Schaefer} if at
least one of the following conditions holds:
\begin{enumerate}
\item every relation in $\mathcal{S}$ is bijunctive.
\item every relation in $\mathcal{S}$ is Horn.
\item every relation in $\mathcal{S}$ is dual Horn.
\item every relation in $\mathcal{S}$ is affine.
\end{enumerate}
A set $\mathcal{S}$ of logical relations is \emph{CPSS} if at least
one of the following conditions holds:
\begin{enumerate}
\item every relation in $\mathcal{S}$ is bijunctive.
\item every relation in $\mathcal{S}$ is Horn and safely componentwise
IHSB$-$.
\item every relation in $\mathcal{S}$ is dual Horn and safely componentwise
IHSB$+$.
\item every relation in $\mathcal{S}$ is affine.
\end{enumerate}
\end{defn}

\section{\label{sec:imp}The Impact of Repeated Variables in Constraints}

In this section we show which of Gopalan et al.'s statements and proofs
in \citep{gop} are affected by the disregard of repeated occurrences
of variables in constraint applications, and how they need to be modified.
In the whole section, we refer to the definitions, lemmas, theorems
and corollaries of \citep{gop}.

\subsection{Application of Structural Expressibility}

The first mistake is in the application of structural expressibility,
defined in \citep{gop} as follows:
\begin{itemize}
\item \noun{Definition 3.1:} A relation $R$ is \emph{structurally expressible}
from a set of relations $\mathcal{S}$ if there is a CNF\textsubscript{C}($\mathcal{S}$)-formula
$\phi$ such that the following conditions hold:

\begin{enumerate}
\item $R=\{\boldsymbol{a}|\exists\boldsymbol{y}\phi(\boldsymbol{a},\boldsymbol{y})\}$.
\item For every $\boldsymbol{a}\in R$, the graph $G(\phi(\boldsymbol{a},\boldsymbol{y}))$
is connected.
\item For $\boldsymbol{a},\boldsymbol{b}\in R$ with $|\boldsymbol{a}-\boldsymbol{b}|=1$,
there exists a \emph{witness} $\boldsymbol{w}$ such that $(\boldsymbol{a},\boldsymbol{w})$
and $(\boldsymbol{b},\boldsymbol{w})$ are solutions of $\phi$.
\end{enumerate}
\end{itemize}
Gopalan et al.~now argue that connectivity were retained for a CNF\textsubscript{C}($\mathcal{S}$)-formula
when replacing in every constraint the applied relation $R$ with
a structural expression of $R$. However, they do not consider how
the connectivity of a relation changes by substitution of constants
and identification of variables. The proof of Lemma 3.2 is only correct
for formulas without constants, where no variable is used more than
once in any constraint (we'll shortly explain why), so we must change
the lemma as follows:
\begin{itemize}
\item \noun{Lemma 3.2:} \emph{Let $\mathcal{S}$ and $\mathcal{S}'$ be
sets of relations such that every $R\in\mathcal{S}'$ is structurally
expressible from $\mathcal{S}$, and, moreover, there is a polynomial-time
algorithm that produces a structural expression from $\mathcal{S}$
for every $R\in\mathcal{S}'$. Given a }CNF($\mathcal{S}'$)\emph{-formula
$\psi(\boldsymbol{x})$ }\textbf{(without constants), where no variable
is used more than once in any constraint}\emph{, one can efficiently
construct a }CNF\textsubscript{C}($\mathcal{S}$)\emph{-formula $\varphi(\boldsymbol{x},\boldsymbol{y})$
such that}

\begin{enumerate}
\item \emph{$\psi(\boldsymbol{x})=\exists\boldsymbol{y}\varphi(\boldsymbol{x},\boldsymbol{y})$;}
\item \emph{if $(\boldsymbol{s}\boldsymbol{,w^{s}}),$ $(\boldsymbol{t},\boldsymbol{w^{t}})$
are connected in $G(\varphi)$ by a path of length $d$, then there
is a path from $\boldsymbol{s}$ to $\boldsymbol{t}$ in $G(\psi)$
of length at most $d$;}
\item \emph{if $\boldsymbol{s},\boldsymbol{t}\in\psi$ are connected in
$G(\psi)$, then for every witness $\boldsymbol{w^{s}}$ of $\boldsymbol{s}$,
and every witness $\boldsymbol{w^{t}}$ of $\boldsymbol{t}$, there
is a path from $(\boldsymbol{s}\boldsymbol{,w^{s}})$ to $(\boldsymbol{t},\boldsymbol{w^{t}})$
in $G(\varphi)$.}
\end{enumerate}
\end{itemize}
In the proof, we only clarify the notation a little:
\begin{itemize}[label= ]
\item ``Let $\psi(\boldsymbol{x})=C_{1}\wedge\cdots\wedge C_{m}$ with
$C_{j}=R_{j}(\boldsymbol{x}_{j})$, where  $R_{j}$ is some relation
from $\mathcal{S}'$, and $\boldsymbol{x}_{j}$ is the vector of variables
to which relation $R_{j}$ is applied\textbf{.} Let $\varphi_{j}$
be the structural expression for $R_{j}$ from ${\cal S}$, so that
$R_{j}(\boldsymbol{x}_{j})\equiv\exists\boldsymbol{y}_{j}~\varphi_{j}(\boldsymbol{x}_{j},\boldsymbol{y}_{j})$.
Let $\boldsymbol{y}$ be the vector $(\boldsymbol{y}_{1},\dots,\boldsymbol{y}_{m})$
and let $\varphi(\boldsymbol{x},\boldsymbol{y})$ be the formula $\wedge_{j=1}^{m}\varphi_{j}(\boldsymbol{x}_{j},\boldsymbol{y}_{j})$.
Then $\psi(\boldsymbol{x})\equiv\exists\boldsymbol{y}~\varphi(\boldsymbol{x},\boldsymbol{y})$.\\
{\small{}\hspace*{4ex}}Statement 2 follows from 1 by projection of
the path on the coordinates of $\boldsymbol{x}$. For statement 3,
consider $\boldsymbol{s},\boldsymbol{t}\in\psi$ that are connected
in $G(\psi)$ via a path $\boldsymbol{s}=\boldsymbol{u^{0}}\rightarrow\boldsymbol{u^{1}}\rightarrow\dots\rightarrow\boldsymbol{u^{r}}=\boldsymbol{t}$
. For every $\boldsymbol{u^{i}},\boldsymbol{u^{i+1}}$, and clause
$C_{j}$, there exists an assignment $\boldsymbol{w_{j}^{i}}$ to
$\boldsymbol{y}_{j}$ such that both $(\boldsymbol{u_{j}^{i}},\boldsymbol{w_{j}^{i}})$
and $(\boldsymbol{u_{j}^{i+1}},\boldsymbol{w_{j}^{i}})$ are solutions
of $\varphi_{j}$, by condition $3$ of structural expressibility.
Thus $(\boldsymbol{u^{i}},\boldsymbol{w^{i}})$ and $(\boldsymbol{u^{i+1}},\boldsymbol{w^{i}})$
are both solutions of $\varphi$, where $\boldsymbol{w^{i}}=(\boldsymbol{w_{1}^{i}},\dots,\boldsymbol{w_{m}^{i}})$.
Further, for every $\boldsymbol{u^{i}}$, the space of solutions of
$\varphi(\boldsymbol{u^{i}},\boldsymbol{y})$ is the product space
of the solutions of $\varphi_{j}(\boldsymbol{u_{j}^{i}},\boldsymbol{y}_{j})$
over $j=1,\dots,m$. Since these are all connected by condition $2$
of structural expressibility, $G(\varphi(\boldsymbol{u^{i}},\boldsymbol{y}))$
is connected. The following describes a path from $(\boldsymbol{s},\boldsymbol{w^{s}})$
to $(\boldsymbol{t},\boldsymbol{w^{t}})$ in $G(\varphi)$: ~$(\boldsymbol{s},\boldsymbol{w^{s}})\rightsquigarrow(\boldsymbol{s},\boldsymbol{w^{0}})\rightarrow(\boldsymbol{u^{1}},\boldsymbol{w^{0}})\rightsquigarrow(\boldsymbol{u^{1}},\boldsymbol{w^{1}})\rightarrow\dots\rightsquigarrow(\boldsymbol{u^{r-1}},\boldsymbol{w^{r-1}})\rightarrow(\boldsymbol{t},\boldsymbol{w^{r-1}})\rightsquigarrow(\boldsymbol{t},\boldsymbol{w^{t}})$.
Here $\rightsquigarrow$ indicates a path in $G(\varphi(\boldsymbol{u^{i}},\boldsymbol{y}))$.''
\end{itemize}
It is easy to show that the statement of this lemma is also correct
if we allow constants in $\psi$; however, we don't need this result.

To see the problem with repeated variables, we have to carefully distinguish
the relation $R\in\mathcal{S}$ used in a constraint of a CNF\textsubscript{C}($\mathcal{S}$)-formula
as ``template'' from the relation $R'$ resulting for the variables
of the formula: Two solutions of $R'$ that differ in only one variable
may originate from solutions of $R$ that differ in more variables,
so that there may be no common witness in a structural expression.
\begin{example}
For a minimal example where connectivity is not retained in a structural
expression when variables are identified in a constraint, consider
$\mathcal{S}=\{Q,R\}$ with $Q=x\vee y$ and $R=\overline{x}\vee\overline{y}$,
and $\mathcal{S}'=\{R'\}$ with $R'=x\vee\overline{y}$, and the structural
expression $R'\equiv\exists z\;Q(x,z)\wedge R(z,y)=\exists z\;(x\vee z)\wedge(\overline{z}\vee\overline{y})$.
Now the CNF\textsubscript{C}($\mathcal{S}'$)-formula $\psi=R'(x,x)=x\vee\overline{x}$
is connected, while $Q(x,z)\wedge R(z,x)=(x\vee z)\wedge(\overline{z}\vee\overline{x})$
is disconnected.

For an example with relevance to the reductions from 3-CNF formulas,
consider Gopalan et al.'s example for a structural expression of the
3-clause $R_{{\rm NAZ}}=\{0,1\}^{3}\setminus\{000\}=x_{1}\vee x_{2}\vee x_{3}$
using the non-tight relation $R_{{\rm NAE}}=\{0,1\}^{3}\setminus\{000,111\}=(x_{1}\vee x_{2}\vee x_{3})\wedge(\overline{x}_{1}\vee\overline{x}_{2}\vee\overline{x}_{3})$
with 
\[
\varphi(x_{1},x_{2},x_{3},y_{1},y_{2})=R_{{\rm {NAE}}}(x_{1},x_{2},y_{1})\wedge R_{{\rm {NAE}}}(x_{2},x_{3},y_{2})\wedge R_{{\rm {NAE}}}(y_{1},y_{2},1):
\]
 Here, $R_{{\rm NAZ}}(x_{1},x_{2},x_{2})=x_{1}\vee x_{2}$ is connected,
while {\small{}
\[
\varphi(x_{1},x_{2},x_{2},y_{1},y_{2})=\left((x_{1}\vee x_{2}\vee y_{1})\wedge(\overline{x}_{1}\vee\overline{x}_{2}\vee\overline{y}_{1})\right)\wedge\left((x_{2}\vee y_{2})\wedge(\overline{x}_{2}\vee\overline{y}_{2})\right)\wedge\left(\overline{y}_{1}\vee\overline{y}_{2}\right)
\]
}is disconnected. See Figure \ref{fig:F}.

\begin{figure}[!h]
\includegraphics[scale=0.6]{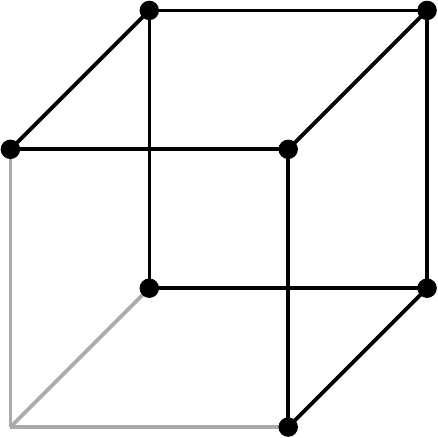}$\quad$\includegraphics[scale=0.6]{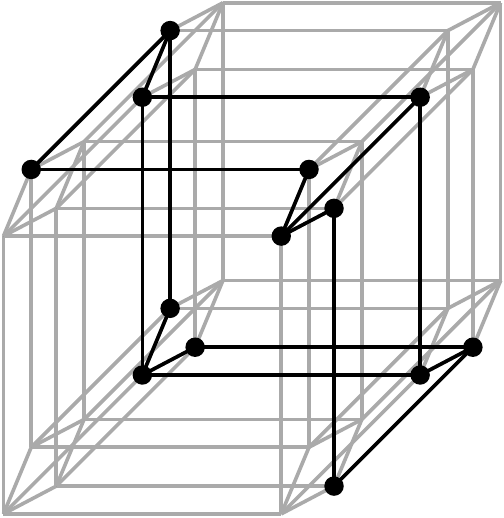}$ $\includegraphics[scale=0.6]{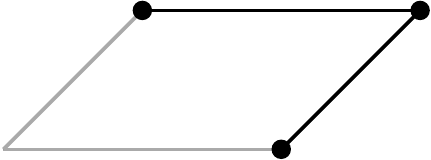}\includegraphics[scale=0.6]{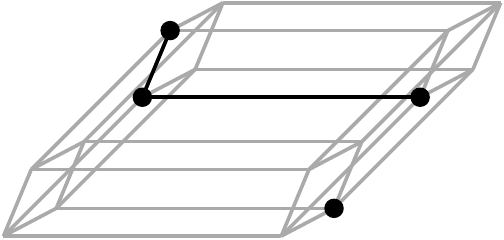}

\protect\caption{\label{fig:F}From left to right: $R_{{\rm NAZ}}(x_{1},x_{2},x_{3})$,
$\varphi(x_{1},x_{2},x_{3},y_{1},y_{2})$, $R_{{\rm NAZ}}(x_{1},x_{2},x_{2})$,
$\varphi(x_{1},x_{2},x_{2},y_{1},y_{2})$. The $x$-coordinates are
plotted along the long coordinate axes, the $y$-coordinates along
the short axes.}
\end{figure}

\end{example}
We have to change Corollary 3.3 accordingly; we denote the connectivity
problems for CNF($\mathcal{S}$)-formulas without repeated variables
in constraints by the subscript $_{\mathrm{ni}}$:
\begin{itemize}
\item \noun{Corollary 3.3:} \emph{Suppose $\mathcal{S}$ and $\mathcal{S}'$
are sets of relations such that every $R\in\mathcal{S}'$ is structurally
expressible from $\mathcal{S}$.}

\begin{enumerate}
\item \emph{There are polynomial-time reductions from }\emph{\noun{Conn}}\emph{\textsubscript{\textbf{ni}}($\mathcal{S}$')
to }\emph{\noun{Conn}}\emph{\textsubscript{C}($\mathcal{S}$), and
from }\emph{\noun{st-Conn}}\emph{\textsubscript{\textbf{ni}}($\mathcal{S}$')
to }\emph{\noun{st-Conn}}\emph{\textsubscript{C}($\mathcal{S}$).}
\item \emph{If there exists a CNF($\mathcal{S}$')-formula $\psi(\boldsymbol{x})$}\textbf{\emph{
where no variable is used more than once in any constraint}}\emph{
with $n$ variables, $m$ clauses and diameter $d$, then there exists
a CNF\textsubscript{C}($\mathcal{S}$)-formula $\phi(\boldsymbol{x},\boldsymbol{y})$,
where $\boldsymbol{y}$ is a vector of $O(m)$ variables, such that
the diameter of $G(\phi)$ is at least $d$.}
\end{enumerate}
\end{itemize}
Corollary 3.3 is used to prove the Theorems 2.8, 2.9 and 2.10, together
with Lemma 3.4, which states that all 3-clauses are structurally expressible
from any non-tight set. With the corollary weakened as above, this
reasoning now is not correct anymore. The solution is to structurally
express the relations resulting for the variables of the formula in
a constraint directly, if variables occur repeatedly and / or constants
are used.

Thus, we now need that not only all 3-clauses, but also all 2- and
all 1-clauses are structurally expressible; however, this follows
from Theorem 2.7, so that the Theorems 2.8, 2.9 and 2.10 are correct
(we will even extend Theorems 2.7, 2.8, 2.9 and 2.10 in the next subsection
to a larger class of relations).

\subsection{Displacement of the Boundaries}

The next mistake is in the generalization of the structural properties
from CNF\textsubscript{C}($\mathcal{S}$)-formulas with bijunctive
sets $\mathcal{S}$ of relations to those with componentwise bijunctive
sets $\mathcal{S}$  in subsection 4.2 of \citep{gop}; the following
mistakes are similar.

As the proofs are stated there, the flaws are quite hard to locate.
The second paragraph of the proof of Lemma 4.3 is supposed to show
that every component of a CNF\textsubscript{C}($\mathcal{S}$)-formulas
$\varphi$ with a set $\mathcal{S}$ of only componentwise bijunctive
relations is the solution space of a formula $\varphi'$ with only
bijunctive relations. To construct $\varphi'$, every constraint $C_{i}$
in $\varphi$ using a relation $R\in\mathcal{S}$ with more than one
component is replaced by a constraint containing only one component
of $R$.

But if in $C_{i}$ some variables of $R$ are identified, and $R$
only is componentwise bijunctive and not safely componentwise bijunctive,
it is possible that the relation resulting for the variables of $\varphi$
is not componentwise bijunctive, as we have seen in \prettyref{exa:ecnp},
and thus not every of its components is bijunctive.

So we must change the lemma:
\begin{itemize}
\item \noun{Lemma 4.3:} \emph{Let }$\mathcal{S}$ \emph{be a set of }\textbf{\emph{safely}}\emph{
componentwise bijunctive relations and $\varphi$ a }CNF\textsubscript{C}($\mathcal{S}$)\emph{-formula.
If $\boldsymbol{a}$ and }\textbf{\emph{$\boldsymbol{b}$}}\emph{
are two solutions of $\varphi$ that lie in the same component of
G($\varphi$), then $d_{\varphi}(\boldsymbol{a},\boldsymbol{b})=|\boldsymbol{a}-\boldsymbol{b}|$,
i.e., no distance expands.}

{\small{}\hspace*{4ex}}In the proof, we replace the second paragraph
by

``For the general case, we show that every component $F$ of G($\varphi$)
is the solution space of a 2-CNF-formula \emph{$\varphi$}. Let $R\in\mathcal{S}$
be a\emph{ }safely componentwise bijunctive relation. Then any relation
corresponding to a clause in $\varphi$ of the form $R(x_{1},\ldots,x_{k})$
(the relation obtained after identifying repeated variables) consists
of bijunctive components $R_{1},\ldots,R_{m}$. The projection of
$F$ onto $x_{1},\ldots,x_{k}$ is itself connected and must satisfy
$R$. Hence it lies within one of the components $R_{1},\ldots,R_{m}$;
assume it is $R_{1}$. We replace $R(x_{1},\ldots,x_{k})$ by $R_{1}(x_{1},\ldots,x_{k})$.
Call this new formula $\varphi_{1}$. $G(\varphi_{1})$ consists of
all components of G($\varphi$) whose projection on $x_{1},\ldots,x_{k}$
lies in $R_{1}$. We repeat this for every clause. Finally we are
left with a formula $\varphi'$ over a set of bijunctive relations.
Hence $\varphi'$ is bijunctive and $G(\varphi')$ is a component
of $G(\varphi)$. So the claim follows from the bijunctive case.''

\end{itemize}
In consequence, the resulting corollary must be changed:
\begin{itemize}
\item \noun{Corollary 4.4:} \emph{Let }$\mathcal{S}$ \emph{be a set of
}\textbf{\emph{safely}}\emph{ componentwise bijunctive relations.
Then...}
\end{itemize}
The proof of Lemma 4.5 is supposed to show by contradiction that every
component of $G(\varphi)$ for a CNF\textsubscript{C}($\mathcal{S}$)-formula
$\varphi$ with a set $\mathcal{S}$ of OR-free relations must contain
a unique locally minimal solution. It is reasoned that if $G(\varphi)$
would contain two locally minimal solutions, the relation corresponding
to some clause $C_{i}$ in $\varphi$ would not be OR-free. This is
correct up to here, with our \prettyref{def:cnf} of ``the relation
corresponding to a clause''. From this it is concluded that some
relation in $\mathcal{S}$ could not have been OR-free. But actually,
$C_{i}$ could have been obtained from an OR-free relation that is
not safely OR-free by identification of variables, as we have seen
in \prettyref{exa:or}. So the lemma must be changed:
\begin{itemize}
\item \noun{Lemma 4.5:} \emph{Let }$\mathcal{S}$ \emph{be a set of }\textbf{\emph{safely}}\emph{
OR-free relations and $\varphi$ a }CNF\textsubscript{C}($\mathcal{S}$)\emph{-formula.
Every component...}

{\small{}\hspace*{4ex}}In the proof, we replace the last sentence
of the second paragraph by ``So the relation corresponding to that
clause is not OR-free, thus \emph{$\mathcal{S}$ }must have contained
some not\emph{ }safely OR-free relation.''

\end{itemize}
In consequence, the resulting corollaries must be changed:
\begin{itemize}
\item \noun{Corollary 4.6:} \emph{Let }$\mathcal{S}$ \emph{be a set of
}\textbf{\emph{safely}}\emph{ OR-free relations. Then...}
\end{itemize}
and
\begin{itemize}
\item \noun{Corollary 4.7:} \emph{Let }$\mathcal{S}$ \emph{be a }\textbf{\emph{safely}}\emph{
tight set of relations. Then...}
\end{itemize}
But now for dichotomies to hold, we must show for every\emph{ }not
safely tight set $\mathcal{S}$  that\emph{ }\noun{Conn}\textsubscript{C}($\mathcal{S}$)
and \emph{\noun{st-}}\noun{Conn}\textsubscript{C}($\mathcal{S}$)
are PSPACE-complete, and that there are CNF\textsubscript{C}($\mathcal{S}$)-formulas
$\varphi$ such that the diameter of $G(\varphi)$ is exponential
in the number of variables of $\varphi$. Therefor, we extend the
structural expressibility theorem (Theorem 2.7) to not safely tight
sets of relations; for this again, Lemma 3.4 must be extended:
\begin{itemize}
\item \noun{Lemma 3.4}: \emph{If set $\mathcal{S}$ of relations is not
}\textbf{\emph{safely}}\emph{ tight, $S_{3}$ is structurally expressible
from $\mathcal{S}$.}

{\small{}\hspace*{4ex}}In the first paragraph of the proof, we replace
``not OR-free'' with ``not \textbf{safely} OR-free'', ``not NAND-free''
with ``not \textbf{safely} NAND-free'', and we express $x_{1}\vee x_{2}$
($\overline{x}_{1}\vee\overline{x}_{2}$) by substitution of constants
\textbf{and identification of variables}. Similarly, in the first
paragraph of ``Step 1'', we replace ``componentwise bijunctive''
by ``\textbf{safely} componentwise bijunctive'', and in the second
paragraph of ``Step 1'' we obtain the required not componentwise
bijunctive relation $R$ from any not safely componentwise bijunctive
relation $R'$ by identification of variables. The remaining part
of the proof need not be modified.

\end{itemize}
Now we can extend the structural expressibility theorem:
\begin{itemize}
\item \noun{Theorem 2.7}: \emph{Let $\mathcal{S}$ be a finite set of logical
relations. If $\mathcal{S}$ is not }\textbf{\emph{safely}}\emph{
tight, then every logical relation is structurally expressible from
$\mathcal{S}$.}
\end{itemize}
Hereby, we can state the dichotomy theorems as follows:
\begin{itemize}
\item \noun{Theorem 2.}8: \emph{Let $\mathcal{S}$ be a finite set of logical
relations. If $\mathcal{S}$ is }\textbf{\emph{safely}}\emph{ tight,
then }\emph{\noun{Conn\textsubscript{C}($\mathcal{S}$)}}\emph{ is
in }\emph{\noun{coNP}}\emph{; otherwise, }\emph{\noun{Conn}}\textsubscript{C}($\mathcal{S}$)\emph{
is }\emph{\noun{PSPACE}}\emph{-complete.}
\item \noun{Theorem 2.}9: \emph{Let $\mathcal{S}$ be a finite set of logical
relations. If $\mathcal{S}$ is }\textbf{\emph{safely}}\emph{ tight,
then }\noun{st-}\emph{\noun{Conn\textsubscript{C}($\mathcal{S}$)}}\emph{
is in }\emph{\noun{P}}\emph{; otherwise, }\noun{st-}\emph{\noun{Conn}}\textsubscript{C}($\mathcal{S}$)\emph{
is }\emph{\noun{PSPACE}}\emph{-complete.}
\item \noun{Theorem 2.}10: \emph{Let $\mathcal{S}$ be a finite set of logical
relations. If $\mathcal{S}$ is }\textbf{\emph{safely}}\emph{ tight,
then} \emph{for every }CNF\textsubscript{C}($\mathcal{S}$)\emph{-formula
$\varphi$, the diameter of $G(\varphi)$ is linear in the number
of variables of $\varphi$; otherwise, there are }CNF\textsubscript{C}($\mathcal{S}$)\emph{-formulas
$\varphi$ such that the diameter of $G(\varphi)$ is exponential
in the number of variables of $\varphi$.}
\end{itemize}
For the inclusion structure of the classes to hold, we now have to
show that all Schaefer sets of relations are safely tight. Therefor,
we tighten\noun{ }Lemma 4.2:
\begin{itemize}
\item \noun{Lemma 4.2}\emph{\noun{:}}\emph{ Let $R$ be a logical relation.}

\begin{enumerate}
\item \emph{If $R$ is bijunctive, then R is }\textbf{\emph{safely}}\emph{
componentwise bijunctive}
\item \emph{If $R$ is Horn, then $R$ is }\textbf{\emph{safely}}\emph{
OR-free.}
\item \emph{If $R$ is dual Horn, then $R$ is }\textbf{\emph{safely}}\emph{
NAND-free.}
\item \emph{If $R$ is affine, then $R$ is }\textbf{\emph{safely}}\emph{
componentwise bijunctive, }\textbf{\emph{safely}}\emph{ OR-free, and
}\textbf{\emph{safely}}\emph{ NAND-free.}
\end{enumerate}

{\small{}\hspace*{4ex}}For the proof, we first note that any relation
obtained from a bijunctive (Horn, dual Horn, affine) one by identification
of variables is itself bijunctive (Horn, dual Horn, affine), which
is obvious from the definitions.\\
Now if in the first case $R$ were componentwise bijunctive but not
safely componentwise bijunctive, there were a relation $R'$ obtained
from $R$ by identification of variables that were bijunctive but
not componentwise bijunctive, which is not possible by the statement
of the original lemma. The reasoning for the other cases is analogous.

\end{itemize}
We have to change Lemma 4.8, since it relies on the wrong assumption
that \noun{Conn}\textsubscript{C}($\mathcal{S}$) is in coNP for
every tight set \emph{$\mathcal{S}$}:
\begin{itemize}
\item \noun{Lemma 4.8:} \emph{For }$\mathcal{S}$ \textbf{\emph{safely}}\emph{
tight, but not Schaefer, }\emph{\noun{Conn\textsubscript{C}($\mathcal{S}$)}}\emph{
is }\emph{\noun{coNP}}\emph{-complete.}

{\small{}\hspace*{4ex}}In the proof, we should clarify that the relation
$x\neq y$ is expressible as a CNF\textsubscript{C}($\mathcal{S}$)-formula,
not necessarily by substitution of constants only, see \prettyref{rem:consi}.

\end{itemize}
Finally, we have to weaken Lemma 4.13. In the last paragraph of the
proof, the connectivity question for a CNF\textsubscript{C}($\mathcal{S}$)-formula
$\varphi$ with a set $\mathcal{S}$ of componentwise IHSB$-$ relations
shall be reduced to one for a formula using only IHSB$-$ relations.
In the last sentence, a false assumption is used: That every relation
corresponding to a clause of $\varphi$ that has only a single component
would be IHSB$-$. Actually, that relation is guaranteed to be IHSB\emph{$-$
}only if the original relation is safely\emph{ }componentwise IHSB$-$,
as we have seen in \prettyref{exa:ecnp}. Thus the lemma must be changed:
\begin{itemize}
\item \noun{Lemma 4.13:} \emph{If }$\mathcal{S}$ \emph{a set of relations
that are Horn (dual Horn) and }\textbf{\emph{safely}}\emph{ componentwise
}IHSB\emph{$-$ (}IHSB\emph{$+$), then there is a polynomial-time
algorithm for}\textbf{\emph{ }}\noun{Conn}\textsubscript{C}($\mathcal{S}$).

{\small{}\hspace*{4ex}}The proof can be retained word-for-word; the
necessary comment ``(the relation obtained after identifying repeated
variables)'' is already mentioned in the last paragraph.

\end{itemize}
The following example shows cases where the above corrections make
a difference.
\begin{example}
Since the relations from Example \ref{exa:ecnp} are Horn and componentwise
IHSB$-$, \emph{\noun{Conn}}$_{C}$($\{R_{\mathrm{coNP}}\}$) and
\emph{\noun{Conn}}$_{C}$($\{[\phi_{\mathrm{coNP}}]\}$) would be
polynomial-time decidable by Lemma 4.13 of \citep{gop}. But $\{R_{\mathrm{coNP}}\}$
and $\{[\phi_{\mathrm{coNP}}]\}$ are not CPSS, thus \emph{\noun{Conn}}$_{C}$($\{R_{\mathrm{coNP}}\}$)
and \emph{\noun{Conn}}$_{C}$($\{[\phi_{\mathrm{coNP}}]\}$) are actually
coNP-complete.

The relation
\[
R_{\mathrm{PSPA}}=\{0001,\,\,0010,\,\,1100,1110,1101\}
\]
 is not Schaefer and not NAND-free ($R_{\mathrm{PSPA}}(1,1,x,y)=\overline{x\wedge y}$),
but componentwise bijunctive and OR-free, thus $\{R_{\mathrm{PSPA}}\}$
is tight, and \emph{\noun{Conn}}$_{C}$($\{R_{\mathrm{PSPA}}\}$)
would be coNP-complete by Lemma 4.8 of \citep{gop}, \emph{\noun{st-Conn}}$_{C}$($\{R_{\mathrm{PSPA}}\}$)
would be polynomial-time decidable by Theorem 2.9 of \citep{gop},
and the diameter of $G(\phi)$ linear in the number of variables for
all CNF\textsubscript{C}($\{R_{\mathrm{PSPA}}\}$)-formulas $\phi$
by Theorem 2.10 of \citep{gop}.

But $R_{\mathrm{PSPA}}$ is not safely componentwise bijunctive (as
identifying the first two variables gives $R'=\{001,010,100,110,101\}$,
and $\mathrm{MAJ}(001,010,100)=000\notin R'$), and not safely OR-free
(as $R'(x,y,0)=x\vee y$), thus $\{R_{\mathrm{PSPA}}\}$ is not safely
tight, and \emph{\noun{Conn}}$_{C}$($\{R_{\mathrm{PSPA}}\}$) and
\emph{\noun{st-Conn}}$_{C}$($\{R_{\mathrm{PSPA}}\}$) are actually
PSPACE-complete, and there are CNF\textsubscript{C}($\{R_{\mathrm{PSPA}}\}$)-formulas
$\phi$ for which the diameter of $G(\phi)$ is exponential in the
number of variables. \end{example}
\begin{rem}
\label{rem:consi}One could of course also consider CNF\textsubscript{C}($\mathcal{S}$)-formulas
without repeated variables in constraints. But in this case, one had
to check all consequences of this restriction. E.g., the proof of
Lemma 4.8 in \citep{gop} were not valid since the relation $x\neq y$
is not expressible without identification of variables from every
non-Schaefer set of relations. For example, for $R=\{1100,1010,1110,0001\}$,
$(x\neq y)$ is $R(x,x,x,y)$, but cannot be obtained from $R$ by
substitution of constants and conjunction only.
\end{rem}

\section{\label{sec:tri}A Trichotomy for \noun{Conn}\protect\textsubscript{C}($\mathcal{S}$)}

In this section, we prove the last piece needed to establish the trichotomy
for \noun{Conn}\textsubscript{C}($\mathcal{S}$).

Initially, Gopalan et al.~conjectured that \noun{Conn}\textsubscript{C}($\mathcal{S}$)
is in P if $\mathcal{S}$ is Schaefer, but this was subsequently disproved
by Makino, Tamaki, and Yamamoto \citep{Makino:2007:BCP:1768142.1768162},
who showed that \noun{Conn}\textsubscript{C}($\mathcal{S}$) is coNP-complete
for $\mathcal{S}=\{x\vee\overline{y}\vee\overline{z}\}$, which is
Horn and thus Schaefer. Consequently, Gopalan et al.~conjectured
that \noun{Conn}\textsubscript{C}($\mathcal{S}$) is coNP-complete
if $\mathcal{S}$ is Horn but not componentwise IHSB$-$, or dual
Horn but not componentwise IHSB$+$, and already suggested a way for
proving that: One had to show that \noun{Conn}$_{C}$($\{M\}$) for
the relation $M=\left(x\vee\overline{y}\vee\overline{z}\right)\wedge\left(\overline{x}\vee z\right)$
is coNP-hard. We will prove this in \prettyref{lem:m har} by a reduction
from the complement of a satisfiability problem.

Gopalan et al.~stated (without giving the proof) they could show
that $M$ is structurally expressible from every set of Horn relations
which contains at least one relation that is not componentwise IHSB$-$,
using a similar reasoning as in the proof of their structural expressibility
theorem. We give a different proof (which may be somewhat simpler)
in \prettyref{lem:exp m}, that shows that $M$ actually is expressible
from every set $\mathcal{S}$ of Horn relations that contains at least
one relation that is not \textbf{safely} componentwise IHSB$-$ as
a CNF\textsubscript{C}($\mathcal{S}$)-formula, which is of course
a structural expression.

In this section, when we refer to results from \citep{gop} that were
corrected or extended in last section, we allude to the modified versions.
\begin{thm}[Trichotomy theorem for \noun{Conn}\textsubscript{C}($\mathcal{S}$)]
Let $\mathcal{S}$ be a finite set of logical relations.
\begin{enumerate}
\item If $\mathcal{S}$ is CPSS,\emph{ }\noun{Conn}\textsubscript{\noun{C}}\noun{($\mathcal{S}$)}
is in \noun{P}.
\item Else, if $\mathcal{S}$ is safely tight, \noun{Conn}\textsubscript{\noun{C}}\noun{($\mathcal{S}$)}
is \noun{coNP}-complete.
\item Else,\emph{ }\noun{Conn}\textsubscript{\noun{C}}\noun{($\mathcal{S}$)}
is \noun{PSPACE}-complete.
\end{enumerate}
\end{thm}
\begin{proof}
1. If $\mathcal{S}$ is CPSS, \noun{Conn}\textsubscript{C}($\mathcal{S}$)
is in P by Lemmas 4.9, 4.13 and 4.10 of \citep{gop}, or by our \prettyref{thm:alg}.

2. If $\mathcal{S}$ is Schaefer and not CPSS, it must be Horn and
contain at least one relation that is not safely componentwise IHSB$-$,
or dual Horn and contain at least one relation that is not safely
componentwise IHSB$+$, and \noun{Conn}\textsubscript{C}($\mathcal{S}$)
is coNP-complete by \prettyref{lem:conp} below. If $\mathcal{S}$
is not Schaefer, the statement follows from Lemma 4.8 of \citep{gop}.

3. This follows from Theorem 2.8 of \citep{gop}.\end{proof}
\begin{lem}
\label{lem:conp}Let $\mathcal{S}$ be a finite set of Horn (dual
Horn) relations. If $\mathcal{S}$ contains at least one relation
that is not safely componentwise IHSB$-$ (not safely componentwise
IHSB$+$), \noun{Conn}\textsubscript{\noun{C}}\noun{($\mathcal{S}$)}
is \noun{coNP}-complete.\end{lem}
\begin{proof}
For sets of Horn relations that contain at least one relation that
is not safely componentwise IHSB$-$, the coNP-hardness follows from
Lemmas \ref{lem:m har} and \ref{lem:exp m} below. The case of sets
of dual Horn relations that contain at least one relation that is
not safely componentwise IHSB$+$ is analogous. Theorem 2.8 of \citep{gop}
shows that \noun{Conn}\textsubscript{C}($\mathcal{S}$) is in coNP.
\end{proof}

\subsection{Connectivity of Horn Formulas}

In this subsection, we introduce terminology and develop tools we
will need for the proofs of \prettyref{lem:m har} and \prettyref{lem:exp m}
in the following two subsections.
\begin{defn}
Clauses with only one literal are called \emph{unit clauses} (\emph{positive}
if the literal is positive, \emph{negative} otherwise). Clauses with
only negative literals are \emph{restraint clauses}, and the sets
of variables occurring in restraint clauses are \emph{restraint sets}.
Clauses having one positive and one or more negative literals are
\emph{implication clauses}. Implication clauses with two or more negative
literals are \emph{multi-implication clauses.}

A variable\emph{ $x$ is implied by a set of variables $U$,} if setting
all variables from $U$ to 1 forces $x$ to be 1 in any satisfying
assignment. We write Imp($U$) for the set of variables implied by
$U$, we abbreviate Imp($\{x\}$) as Imp($x$). We simply say that
$x$ \emph{is} \emph{implied,} if $x\in\mathrm{Imp}(U\setminus\{x\})$
for some $U$. Note that $U\subseteq\mathrm{Imp}(U)$ for all sets
$U$.

$U$ is \emph{self-implicating} if every $x\in U$ is implied by $U\setminus\{x\}$.
$U$ is \emph{maximal self-implicating}, if further $U=\mathrm{Imp}(U)$.\end{defn}
\begin{rem}
\label{hyp}A Horn formula can be represented by a directed hypergraph
with hyperedges of head-size one as follows: For every variable, there
is a node, for every implication clause $y\vee\overline{x}_{1}\vee\cdots\vee\overline{x}_{k}$,
there is a directed hyperedge from $x_{1},\ldots,x_{k}$ to $y$,
for every restraint clause $\overline{x}_{1}\vee\cdots\vee\overline{x}_{k}$,
there is a directed hyperedge from $x_{1},\ldots,x_{k}$ to a special
node labeled ``false'', and for every positive unit clause $x$,
there is a directed hyperedge from a special node labeled ``true''
to $x$. For simplicity, we omit the ``false'' and ``true'' nodes
in the drawings and let the corresponding hyperedges end, resp. begin,
in the void.

We draw the directed hyperedges as joining lines, e.g., $x\vee\overline{y}\vee\overline{z}=$
$\vcenter{\hbox{\includegraphics[scale=0.33]{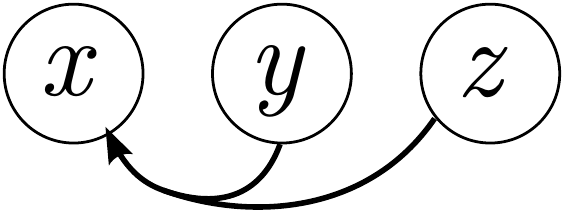}}}$.
\end{rem}
Further, we will use some terminology from \citep{gop}:
\begin{defn}
\label{def:lm}The \emph{coordinate-wise partial order} $\leq$ on
Boolean vectors is defined as follows: $\boldsymbol{a}\leq\boldsymbol{b}$
iff $a_{i}\leq b_{i}$ $\forall i$. A \emph{monotone path} between
two solutions $\boldsymbol{a}$ and $\boldsymbol{b}$ is a path $\boldsymbol{a}\rightarrow\boldsymbol{u}^{1}\rightarrow\cdots\rightarrow\boldsymbol{u}^{r}\rightarrow\boldsymbol{b}$
in the solution graph such that $\boldsymbol{a}\leq\boldsymbol{u}^{1}\leq\cdots\leq\boldsymbol{u}^{r}\leq\boldsymbol{b}$.
A solution is \emph{locally minimal} iff it has no neighboring solution
that is smaller than it.\end{defn}
\begin{lem}
\label{lem:loc min}The solution graph of a Horn formula $\phi$ without
positive unit clauses is disconnected iff $\phi$ has a locally minimal
nonzero solution.\end{lem}
\begin{proof}
This follows from Lemma 4.5 of \citep{gop} since the all-zero vector
is a solution of every Horn formula without positive unit clauses,
and Horn formulas are safely OR-free by Lemma 4.2 of \citep{gop}.\end{proof}
\begin{lem}
\label{lem:horn conn}For every Horn formula $\phi$ without positive
unit clauses, there is a bijection correlating each connected component
$\phi_{i}$ with a maximal self-implicating set $U_{i}$ containing
no restraint set; $U_{i}$ consists of the variables assigned 1 in
the minimum solution of $\phi_{i}$ (the ``lowest'' component is
correlated with the empty set).\end{lem}
\begin{proof}
Let $\phi_{i}$ be a connected component of $\phi$ with minimum solution
$\boldsymbol{s}$, and let $U$ be the set of variables assigned 1
in $\boldsymbol{s}$. Since $\boldsymbol{s}$ is locally minimal,
flipping any variable $x_{i}$ from $U$ to 0 results in a vector
that is no solution, so there must be a clause in $\phi$ prohibiting
that $x_{i}$ is flipped. Since $\phi$ contains no positive unit-clauses,
each $x_{i}\in U$ must appear as the positive literal in an implication
clause with also all negated variables from $U$. It follows that
$U$ is self-implicating. Also, $U$ must be maximal self-implicating
and can contain no restraint set, else $\boldsymbol{s}$ were no solution.

Conversely, let $U$ be a maximal self-implicating set containing
no restraint set. Then the vector $\boldsymbol{s}$ with all variables
from $U$ assigned 1, and all others 0, is a locally minimal solution:
All implication clauses $\overline{y}_{1}\vee\cdots\vee\overline{y}_{k}\vee x$
with some $y_{i}\notin U$ are satisfied since $y_{i}=0$, and for
the ones with all $y_{i}\in U$, also $x\in U$ holds because $U$
is maximal, so these are satisfied since $x=1$. All restraint clauses
are satisfied since $U$ contains no restraint set. $\boldsymbol{s}$
is locally minimal since every variable assigned 1 is implied by $U$,
so that any vector with one such variable flipped to 0 is no solution.
By Lemma 4.5 of \citep{gop}, every connected component has a unique
locally minimal solution, so $\boldsymbol{s}$ is the minimum solution
of some component.\end{proof}
\begin{cor}
\label{cor:horn conn}The solution graph of a Horn formula $\phi$
without positive unit clauses is disconnected iff $\phi$ has a non-empty
maximal self-implicating set containing no restraint set.
\end{cor}

\begin{lem}
\label{lem:in}Let $R_{1}$ and $R_{2}$ be two connected components
of a Horn relation $R$ with minimum solutions $\boldsymbol{u}$ and
$\boldsymbol{v}$, resp., and let $U$ and $V$ be the sets of variables
assigned 1 in $\boldsymbol{u}$ and $\boldsymbol{v}$, resp. If then
$U\subsetneq V$, no vector $\boldsymbol{a}\in R_{1}$ has all variables
from $V$ assigned 1.\end{lem}
\begin{proof}
For the sake of contradiction, assume $\boldsymbol{a}\in R_{1}$ has
all variables from $V$ assigned 1. Then $\boldsymbol{a\wedge v}=\boldsymbol{v}$,
where $\boldsymbol{\wedge}$ is applied coordinate-wise. Consider
a path from $\boldsymbol{u}$ to $\boldsymbol{a}$, $\boldsymbol{u}\rightarrow\boldsymbol{w}^{1}\rightarrow\cdots\rightarrow\boldsymbol{w}^{k}\rightarrow\boldsymbol{a}$.
Since $U\subsetneq V$, we have $\boldsymbol{u\wedge v}=\boldsymbol{u}$,
so we can construct a path from $\boldsymbol{u}$ to $\boldsymbol{v}$
by replacing each $\boldsymbol{w}^{i}$ by $\boldsymbol{w}^{i}\boldsymbol{\wedge}\boldsymbol{v}$
in the above path, and removing repetitions. Since $R$ is Horn, it
is closed under $\wedge$ (see \prettyref{lem:cl}), so all vectors
of the constructed path are in $R$. But $\boldsymbol{u}$ and $\boldsymbol{v}$
are not connected in $R$, which is a contradiction.
\end{proof}

\subsection{Reduction from Satisfiability}
\begin{lem}
\noun{\label{lem:m har}Conn}\textsubscript{\noun{C}}\noun{(}\textup{$\left\{ \left(x\vee\overline{y}\vee\overline{z}\right)\wedge\left(\overline{x}\vee z\right)\right\} $}\noun{)}
is \noun{coNP}-hard.\end{lem}
\begin{proof}
\begin{figure}[!h]
\begin{centering}
\includegraphics[scale=0.39]{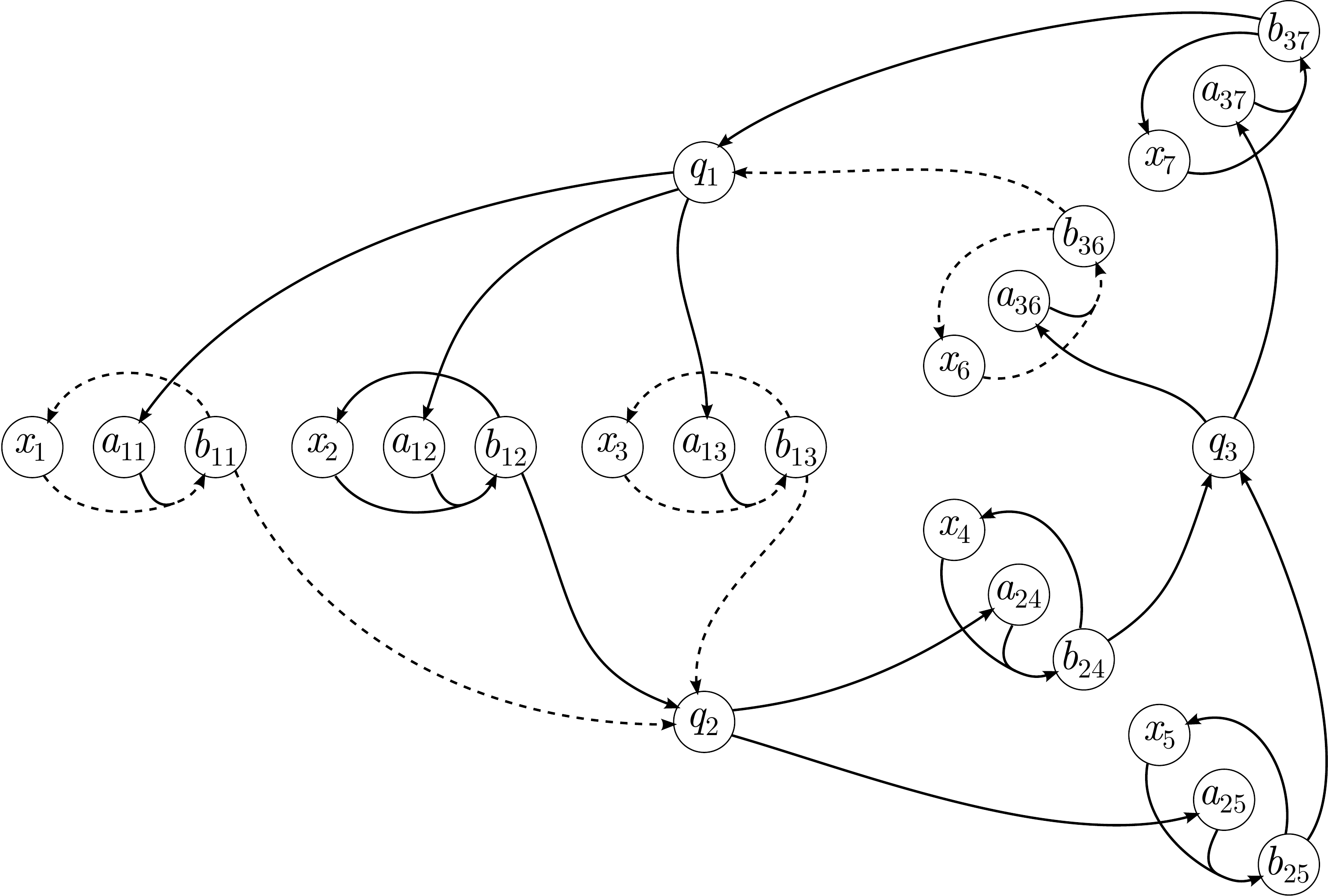}
\par\end{centering}

\protect\caption{\label{fig:cir}\setlength{\parindent}{1em}An example for the proof
of \prettyref{lem:m har}, illustrating the idea. Depicted here is
the hypergraph representation (see Remark \ref{hyp}) of $\phi$ for
$\psi=\left(x_{1}\vee x_{2}\vee x_{3}\right)\wedge\left(x_{4}\vee x_{5}\right)\wedge\left(x_{6}\vee x_{7}\right)$,
as constructed in the proof.\protect \\
\hspace*{4ex}Any self-implicating set of $\phi$ must contain a ``large
circulatory'', passing through each $q_{p}$ and at least one gadget
$\vcenter{\hbox{\protect\includegraphics[scale=0.36]{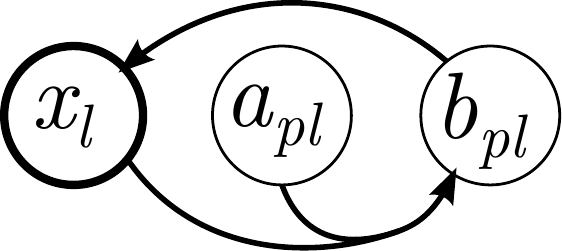}}}$ for
each $p$; these gadgets act as ``valves'': If some $x_{i}$ is
not allowed to be assigned 1 (due to restraint clauses), the circulatory
may not pass through any gadget containing $x_{i}$.\protect \\
\hspace*{4ex}Every maximal self-implicating set also contains all
$a_{pl}$; here, for example, one maximal self-implicating set consist
of the variables with the outgoing edges drawn solid.\protect \\
\hspace*{4ex}If we would add restraint clauses to $\psi$ s.t. $\psi$
would become unsatisfiable, e.g. $\overline{x_{4}}\vee\overline{x_{6}}$,
$\overline{x_{4}}\vee\overline{x_{7}}$, $\overline{x_{5}}\vee\overline{x_{6}}$,
and $\overline{x_{5}}\vee\overline{x_{7}}$, each maximal self-implicating
set of the corresponding $\phi$ would contain a restraint set, so
that $G(\phi)$ would be connected.}
\end{figure}

\begin{figure}[!h]
\begin{centering}
\includegraphics[scale=0.42]{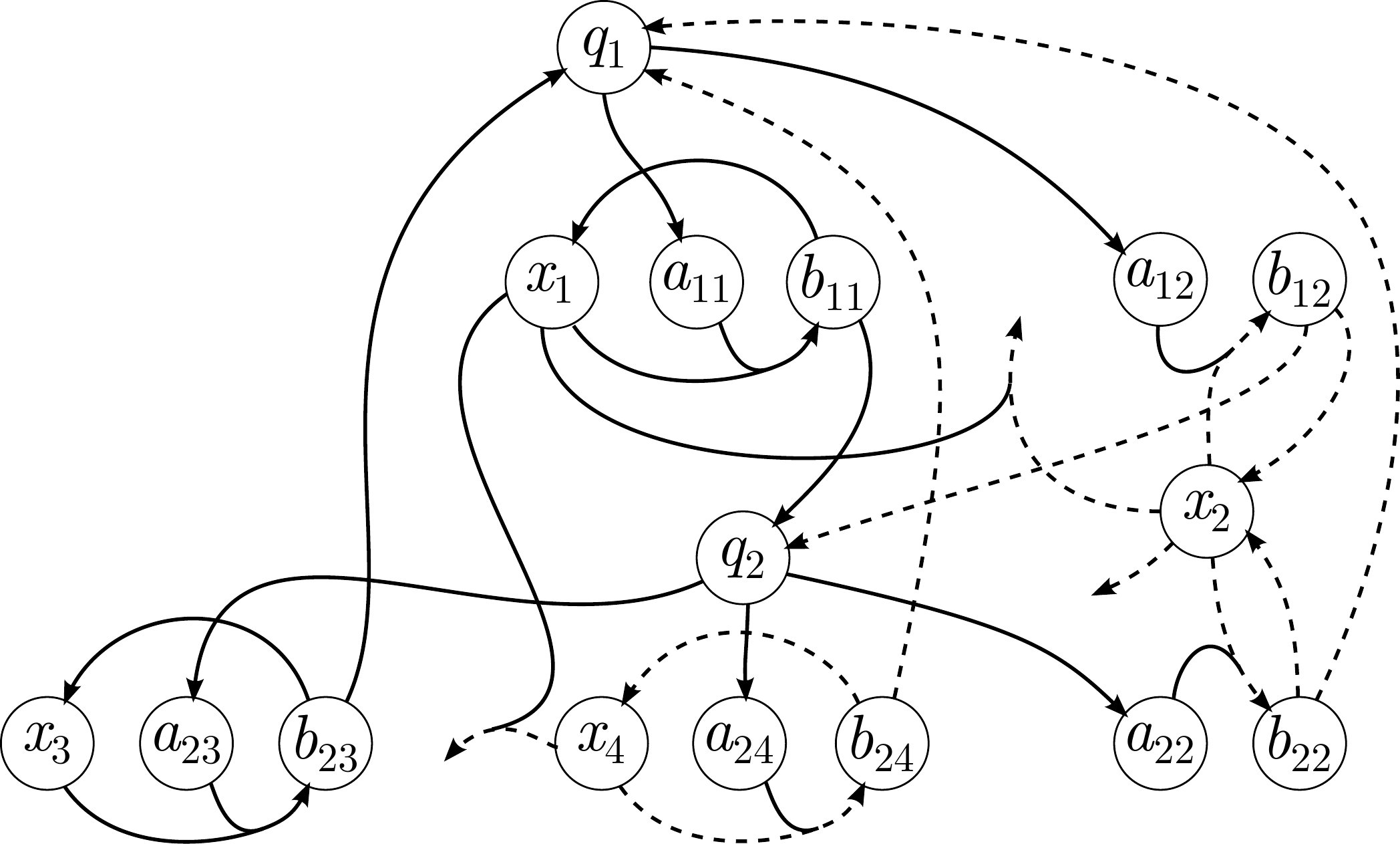}
\par\end{centering}

\protect\caption{\label{fig:gr}A more complex example, with a variable of $\psi$
appearing twice in a $P$-constraint: Depicted is $\phi$ for $\psi=\left(x_{1}\vee x_{2}\right)\wedge\left(x_{3}\vee x_{4}\vee x_{2}\right)\wedge\left(\overline{x}_{1}\vee\overline{x}_{2}\right)\wedge\left(\overline{x}_{1}\vee\overline{x}_{4}\right)\wedge\left(\overline{x}_{2}\right)$.
\protect \\
\hspace*{4ex}$\psi$ is satisfiable with the unique solution $x_{1}=x_{3}=1$
and $x_{2}=x_{4}=0$, and $G(\phi)$ is disconnected (with exactly
two components, since there is exactly one maximal self-implicating
set containing no restraint set, consisting of the variables with
the outgoing edges drawn solid).}
\end{figure}

We reduce the no-constants satisfiability problem \noun{Sat}($\left\{ P,N\right\} $)
with $P=x\vee y\vee z$ and $N=\overline{x}\vee\overline{y}$ to the
complement of \noun{Conn}$_{C}$($\left\{ M\right\} $), where $M=\left(x\vee\overline{y}\vee\overline{z}\right)\wedge\left(\overline{x}\vee z\right)$.
\noun{Sat}($\left\{ P,N\right\} $) is NP-hard by Schaefer's dichotomy
theorem (Theorem 2.1 in \citep{Schaefer:1978:CSP:800133.804350})
since $P$ is not 0-valid, not bijunctive, not Horn and not affine,
while $N$ is not 1-valid and not dual Horn.

Let $\psi$ be any CNF($\left\{ P,N\right\} $)-formula. If $\psi$
only contains $N$-constraints, it is trivially satisfiable, so assume
it contains at least one $P$-constraint. We construct a CNF\textsubscript{C}($\{M\}$)-formula
$\phi$ s.t.$\mbox{\,}$the solution graph $G(\phi)$ is disconnected
iff $\psi$ is satisfiable. First note that we can use the relations
$\overline{x}\vee\overline{y}=M(0,x,y)$ and $\overline{x}\vee y=M(x,0,y)$. 

For every variable $x_{i}$ of $\psi$ ($i=1,\ldots,n$), there is
the same variable $x_{i}$ in $\phi$. For every $N$-constraint $\overline{x}_{i}\vee\overline{x}_{j}$
of $\psi$, there is the clause $\overline{x}_{i}\vee\overline{x}_{j}$
in $\phi$ also. For every $P$-constraint $c_{p}=x_{i_{p}}\vee x_{j_{p}}\vee x_{k_{p}}$
($p=1,\ldots,m$) of $\psi$ there is an additional variable $q_{p}$
in $\phi$, and for every $x_{l}\in\{x_{i_{p}},x_{j_{p}},x_{k_{p}}\}$
appearing in $c_{p}$, there are two more additional variables $a_{pl}$
and $b_{pl}$ in $\phi$. Now for every $c_{p}$, for each $l\in\{i_{p},j_{p},k_{p}\}$
the constraints $\overline{q}_{p}\vee a_{pl}$, $\left(\overline{x}_{l}\vee\overline{a}_{pl}\vee b_{pl}\right)\wedge\left(\overline{b}_{pl}\vee x_{l}\right)$
and $\overline{b}_{pl}\vee q_{(p+1)\,\mathrm{mod}\,m}$ are added
to $\phi$. See the figures for examples of the construction. 

If $\psi$ is satisfiable, there is an assignment $\boldsymbol{s}$
to the variables $x_{i}$ s.t.$\mbox{\,}$for every $P$-constraint
$c_{p}$ there is at least one $x_{l}\in\{x_{i_{p}},x_{j_{p}},x_{k_{p}}\}$
assigned 1, and for no $N$-constraint $\overline{x}_{i}\vee\overline{x}_{j}$,
both $x_{i}$ and $x_{j}$ are assigned 1. We extend $\boldsymbol{s}$
to a locally minimal nonzero satisfying assignment $\boldsymbol{s}'$
for $\phi$; then $G(\phi)$ is disconnected by \prettyref{lem:loc min}:
Let all $q_{p}=1$, $a_{pl}=1$, and all $b_{pl}=x_{l}$ in $\boldsymbol{s}'$.
It is easy to check that all clauses of $\phi$ are satisfied, and
that all variables assigned 1 appear as the positive literal in an
implication clause with all its variables assigned 1, so that $\boldsymbol{s}'$
is locally minimal. $\boldsymbol{s}'$ is nonzero since $\psi$ contains
at least one $P$-constraint.

Conversely, if $G(\phi)$ is disconnected, $\phi$ has a maximal self-implicating
set $U$ containing no restraint set by \prettyref{cor:horn conn}.
It is easy to see that $U$ must contain all $q_{p}$, all $a_{pl}$,
and for every $p$ for at least one $l\in\{i_{p},j_{p},k_{p}\}$ both
$b_{pl}$ and $x_{l}$ (see also Figure \ref{fig:cir} and the explanation
beneath). Thus the assignment with all $x_{i}\in U$ assigned 1 and
all other $x_{i}$ assigned 0 satisfies $\psi$.
\end{proof}

\subsection{Expressing \emph{M}}
\begin{lem}
\label{lem:exp m}The relation $M=\left(x\vee\overline{y}\vee\overline{z}\right)\wedge\left(\overline{x}\vee z\right)$
is expressible as a \noun{CNF\textsubscript{C}($\{R\}$)}-formula
for every Horn relation $R$ that is not safely componentwise IHSB$-$.\end{lem}
\begin{proof}
$M=\left(x\vee\overline{y}\vee\overline{z}\right)\wedge\left(\overline{x}\vee z\right)=$
$\vcenter{\hbox{\includegraphics[scale=0.33]{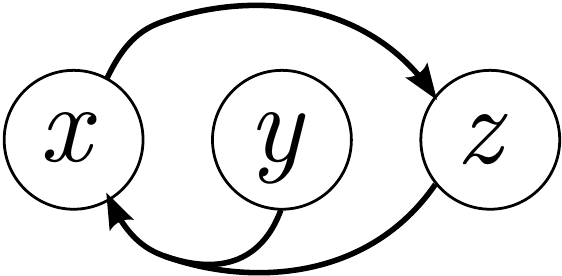}}}$ contains a multi-implication
clause where some negated variable is not implied. The only other
3-ary such relations are (up to permutation of variables)

$L=\left(x\vee\overline{y}\vee\overline{z}\right)\wedge\left(\overline{x}\vee\overline{y}\vee z\right)=$
$\vcenter{\hbox{\includegraphics[scale=0.33]{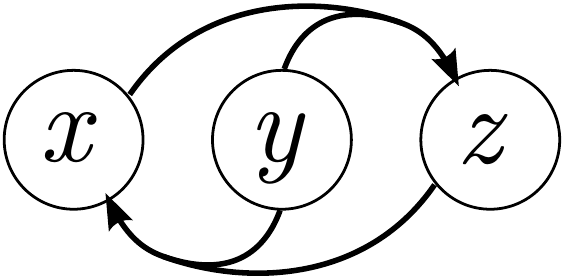}}}$$\enskip$and
$K=x\vee\overline{y}\vee\overline{z}=$ $\vcenter{\hbox{\includegraphics[scale=0.33]{drws}}}$.
We show that $M,L,$ or $K$ is expressible from $R$ by substitution
of constants and identification of variables. We can then express
$M$ from $K$ or $L$ as 
\[
M(x,y,z)\equiv K(x,y,z)\wedge K(z,x,x)\equiv L(x,y,z)\wedge L(z,x,x).
\]

We will argue with formulas in a certain normal form $\nu$; for a
Horn formula $\phi$, let $\nu(\phi)$ be the formula obtained from
$\phi$ by recursively applying the following simplification rules
as long as one is applicable; it is easy to check that the operations
are equivalent transformations, and that the recursion must terminate:
\selectlanguage{british}%
\begin{enumerate}[label=({\alph*})]
\item \foreignlanguage{english}{The constants 0 and 1 are eliminated in
the obvious way.}
\selectlanguage{english}%
\item Multiple occurrences of some variable in a clause are eliminated in
the obvious way.
\item \emph{\label{enu:red imp}}If for some implication clause $c=x\vee\overline{y}_{1}\vee\cdots\vee\overline{y}_{k}$
($k\geq1$), $x$ is already implied by $\{y_{1},\ldots,y_{k}\}$
via other clauses, $c$ is removed.\emph{}\\
{\small{}\hspace*{3ex}}E.g., if there was a clause $x\vee\overline{z}_{1}\vee\cdots\vee\overline{z}_{l}$
with $\{z_{1},\ldots,z_{l}\}\subseteq\{y_{1},\ldots,y_{k}\}$, or
clauses $q\vee\overline{z}_{1}\vee\cdots\vee\overline{z}_{l}$ and
$x\vee\overline{q}$, \emph{$c$} would be removed. Which clauses
are removed by this rule may be random; e.g., for the formula $\left(\overline{x}\vee y\right)\wedge\left(\overline{x}\vee z\right)\wedge\left(\overline{z}\vee y\right)\wedge\left(\overline{y}\vee z\right)$,
$\overline{x}\vee y$ \emph{or} $\overline{x}\vee z$ would be removed:\\
\includegraphics[scale=0.42]{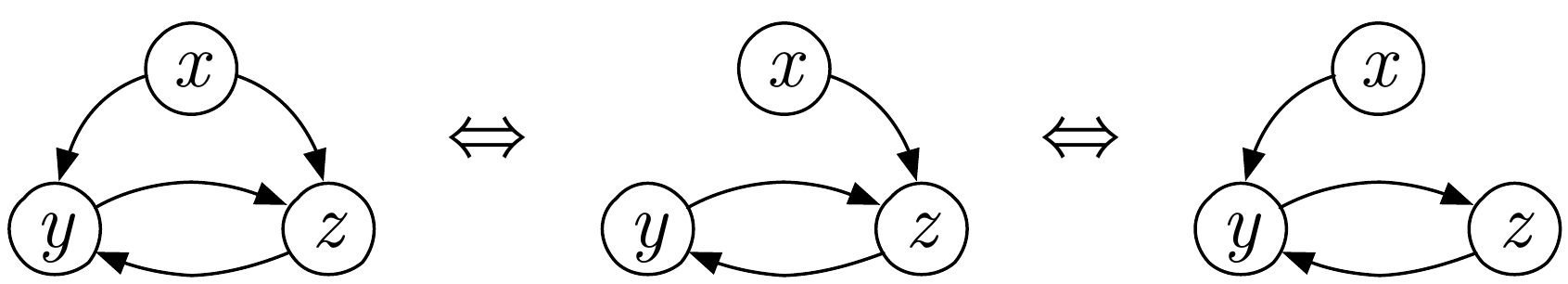}
\item \emph{\label{enu:imp imp}}If for some implication clause $c=x\vee\overline{y}_{1}\vee\cdots\vee\overline{y}_{k}$
($k\geq1$), Imp(Var($c$)) contains a restraint set, $c$ is replaced
by $\overline{y}_{1}\vee\cdots\vee\overline{y}_{k}$.\\
{\small{}\hspace*{3ex}}E.g., if there is a clause $\overline{r}_{1}\vee\cdots\vee\overline{r}_{l}$
with $\{r_{1},\ldots,r_{l}\}\subseteq\{x,y_{1},\ldots,y_{k}\}$, or
if there are clauses $q_{1}\vee\overline{r}_{1}\vee\cdots\vee\overline{r}_{l}$,
$q_{2}\vee\overline{r}_{1}\vee\cdots\vee\overline{r}_{l}$ and $\overline{q}_{1}\vee\overline{q}_{2}$.
E.g., in the formula $\left(x\vee\overline{y}\vee\overline{z}\right)\wedge\left(w\vee\overline{y}\right)\wedge\left(\overline{w}\vee\overline{x}\right)$,
$x\vee\overline{y}\vee\overline{z}$ is replaced by $\overline{y}\vee\overline{z}$:\\
\includegraphics[scale=0.42]{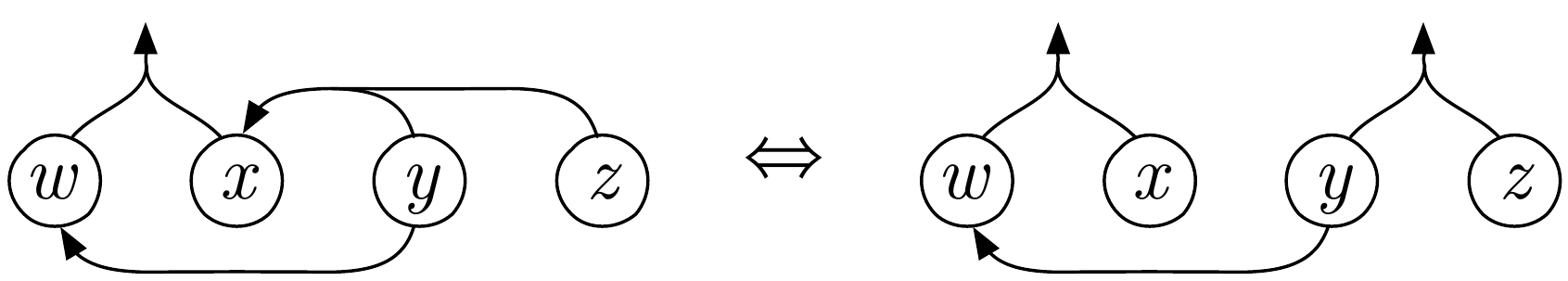}
\item \emph{\label{enu:red bra}}If for some multi-implication clause $c=x\vee\overline{y}_{1}\vee\cdots\vee\overline{y}_{k}$
($k\geq2$), or for some restraint clause $d=\overline{y}_{1}\vee\cdots\vee\overline{y}_{k}$,
some $y_{i}\in\{y_{1},\ldots,y_{k}\}$ is implied by $\{y_{1},\ldots,y_{k}\}\setminus\{y_{i}\}$,
the literal $\overline{y}_{i}$ is removed from $c$ resp. $d$.\emph{}\\
{\small{}\hspace*{3ex}}Which literals are removed by this rule may
be random, as in the following example:\\
\includegraphics[scale=0.47]{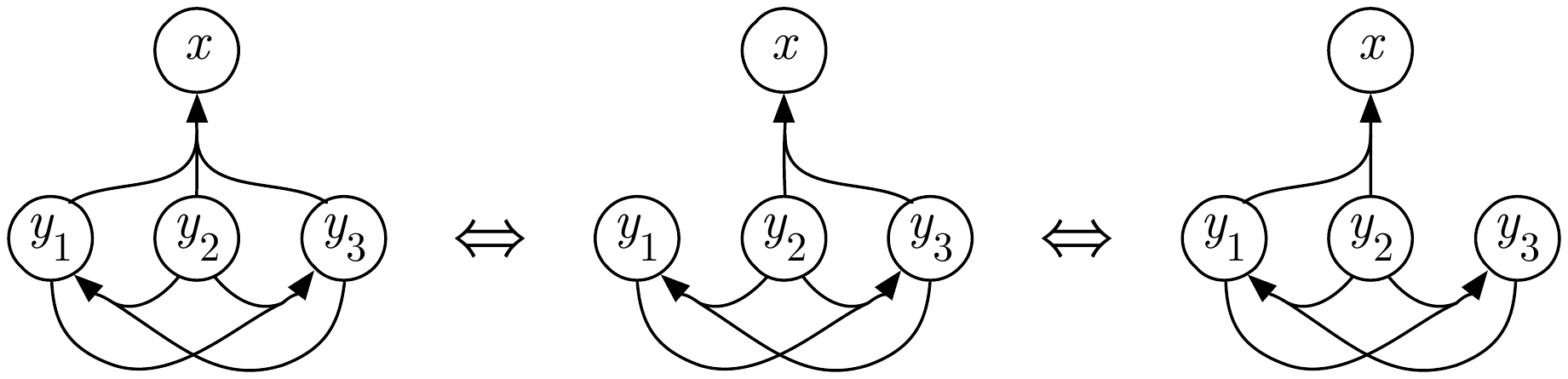}
\end{enumerate}
\selectlanguage{english}%
Let $\phi_{0}$ be a formula representing $R$, in normal form $\nu$
. The following 7 numbered transformation steps generate $K$, $L$,
or $M$ from $\phi_{0}$. After each transformation, we assume the
resulting formula is simplified to be in normal form; we denote the
formula resulting from the $i$'th transformation step in this way
by $\phi_{i}$.

In the first three steps, we ensure that the formula contains a multi-implication
clause where some variable is not implied, in the fourth step we trim
the multi-implication clause to size 3, and in the last three steps
we eliminate all remaining clauses and variables not occurring in
$K$, $L$, or $M$. Our first goal is to produce a formula with a
connected solution graph that is not IHSB$-$, which will turn out
helpful.
\begin{enumerate}
\item \emph{Obtain a not componentwise IHSB$-$ formula $\phi_{1}$ from
$\phi_{0}$ by identification of variables.}
\end{enumerate}
Let $\left[\phi_{1}^{*}\right]$ be a connected component of $\left[\phi_{1}\right]$
that is not IHSB$-$, and let $U$ be the set of variables assigned
1 in the minimum solution of $\phi_{1}^{*}$.
\begin{enumerate}[resume]
\item \emph{\label{enu:s1}Substitute 1 for all variables from $U$.}
\end{enumerate}
The resulting formula $\phi_{2}$ now contains no positive unit-clauses.
Further, the component $\left[\phi_{2}^{*}\right]$ of $\left[\phi_{2}\right]$
resulting from $\left[\phi_{1}^{*}\right]$ is still not IHSB$-$,
and it has the all-0 vector as minimum solution. We show that 
\begin{equation}
{\textstyle \phi_{2}^{*}\equiv\nu\left(\phi_{2}\wedge\left(\bigvee_{x\in V_{1}}\overline{x}\right)\wedge\cdots\wedge\left(\bigvee_{x\in V_{k}}\overline{x}\right)\right),}\label{eq:st}
\end{equation}
where $V_{1},\ldots,V_{k}$ are the sets of variables assigned 1 in
the minimum solutions $\boldsymbol{v}^{1},\ldots,\boldsymbol{v}^{k}$
of the other components of $\left[\phi_{2}\right]$, and we specified
the formula to be in normal form:
\begin{itemize}[label= ]
\item For any solution $\boldsymbol{a}$ in the component with minimum
solution $\boldsymbol{v}^{i}$ we have $\boldsymbol{a}\geq\boldsymbol{v}^{i}$
(see \prettyref{def:lm}), so all components other than $\left[\phi_{2}^{*}\right]$
are eliminated in the right-hand side of \eqref{eq:st}. By \prettyref{lem:in},
no vector from $\left[\phi_{2}^{*}\right]$ is removed.
\end{itemize}
By \prettyref{lem:horn conn}, $V_{1},\ldots,V_{k}$ are exactly the
non-empty maximal self-implicating sets of $\phi_{2}$ that contain
no restraint set.

Clearly, $\phi_{2}^{*}$ is not IHSB$-$. However, we have no restraint
clauses at our disposal to generate $\phi_{2}^{*}$ from $\phi_{2}$;
nevertheless, we can isolate a connected part of $\phi_{2}^{*}$ that
is not IHSB$-$, as we will see.

Since $\phi_{2}^{*}$ is not IHSB$-$, it contains a multi-implication
clause $c^{*}$, and by \eqref{eq:st} it is clear that $\phi_{2}$
must contain the same clause $c^{*}$.

By simplification rule \ref{enu:imp imp}, Imp(Var$(c^{*})$) contains
no restraint set in $\phi_{2}$. Now if some self-implicating set
$U^{*}$ were implied by Var$(c^{*})$, the related maximal self-implicating
set $U_{m}^{*}$ (which then were also implied by Var$(c^{*})$) could
contain no restraint set, thus a restraint clause would be added for
the variables from $U_{m}^{*}$ in \eqref{eq:st}. But then $c^{*}$
would be removed by $\nu$ in \eqref{eq:st}, again due to rule \ref{enu:imp imp},
which is a contradiction. Thus Imp(Var$(c^{*})$) also contains no
self-implicating set in $\phi_{2}$, and so the following operation
eliminates all self-implicating sets and all restraint clauses:
\begin{enumerate}[resume]
\item \emph{Substitute 0 for all remaining variables not implied by} Var($c^{*}$).
\end{enumerate}
This operation also produces no new restraint clauses since any implication
clause with the positive literal not implied by Var($c^{*}$) must
also have some negative literal not implied by Var($c^{*}$), and
thus vanishes.

Further, since $\phi_{2}$ contained no positive unit-clauses, the
formula cannot have become unsatisfiable by this operation. Also,
it is easy to see that the simplification initiated by the substitution
of 0 for some variable $x_{i}$ can only affect clauses $c$ with
$x_{i}\in\mathrm{Imp}(\mathrm{Var}(c))$, so $c^{*}$ is retained
in $\phi_{3}$. 

Since all variables not from Var($c^{*}$) are now implied by\emph{
}Var($c^{*}$), and Imp(Var($c^{*}$)) is not self-implicating, $c^{*}$
contains a variable that is not implied; w.l.o.g., let $c^{*}=x\vee\overline{y}\vee\overline{z}_{1}\vee\cdots\vee\overline{z}_{k}$
($k\geq1$) s.t. $y$ is not implied.
\begin{enumerate}[resume]
\item \emph{\label{enu:Id z}Identify $z_{1},\ldots,z_{k}$, call the resulting
variable $z$.}
\end{enumerate}
This produces the clause $c^{\sim}=x\vee\overline{y}\vee\overline{z}$
from $c^{*}$. Clearly, $y$ is still not implied in $\phi_{4}$,
and since $x$ was not implied by any set $U\subsetneq\{y,z_{1},\ldots,x_{k}\}$
by simplification rule \ref{enu:red bra} in $\phi_{3}$, and no $z_{i}$
was implied by $y$, it follows for $\phi_{4}$ that
\begin{itemize}[label= ({*})]
\item $x\notin\mathrm{Imp}(y)$, $x\notin\mathrm{Imp}(z)$, $z\notin\mathrm{Imp}(y)$,
$y$ is not implied.
\end{itemize}
Also, since $x$ was implied by $\{y,z_{1},\ldots,x_{k}\}$ only via
$c^{*}$ in $\phi_{3}$ due to simplification rule \ref{enu:red imp},
$x$ is implied by $\{y,z\}$ only via $c^{\sim}$ in $\phi_{4}$.

In the following steps, we eliminate all variables other than $x,y,z,$
s.t. $c^{\sim}$ is retained and ({*}) is maintained. It follows that
we are then left with $K,L$, or $M$, since the only clauses only
involving $x,y,z$ and satisfying ({*}) besides $c^{\sim}$ are from
$\{z\vee\overline{x},\,z\vee\overline{x}\vee\overline{y}\}$.
\begin{enumerate}[resume]
\item \emph{Substitute 1 for every variable from $\mathrm{Imp}(y)\setminus\{y\}$
. }
\end{enumerate}
For the simplification initiated by this operation, note that $\phi_{4}$
contained no restraint clauses. It follows that the formula cannot
have become unsatisfiable by this operation. Further, it is easy to
see that for a Horn formula without restraint clauses, at a substitution
of 1 for variables from a set $U$, only clauses $c$ containing at
least one variable $x_{i}\in\mathrm{Imp}(U)$ are affected by the
simplification. Thus, $c^{\sim}$ is not affected since $x,y$ and
$z$ were not implied by $\mathrm{Imp}(y)\setminus\{y\}$.

We must carefully check that ({*}) is maintained since substitutions
of 1 may result in new implications: Since $\mathrm{Imp}(y)\setminus\{y\}$
is empty in $\phi_{5}$, still $x\notin\mathrm{Imp}(y)$ and $z\notin\mathrm{Imp}(y)$.
It is easy to see that $x$ could only have become implied by $z$
as result of transformation 5 if there had been a multi-implication
clause (other than $c^{\sim}$) in $\phi_{4}$ with the positive variable
implying $x$, and each negated variable implied by $y$ or $z$;
but this is not the case since $x$ was implied by $\{y,z\}$ only
via $c^{\sim}$ in $\phi_{4}$, thus still $x\notin\mathrm{Imp}(z)$.

We eliminate all remaining variables besides $x,y,z$ by identifications
in the next two steps. Since now $\mathrm{Imp}(y)\setminus\{y\}$
is empty, the only condition from ({*}) we have to care about is that
$x\notin\mathrm{Imp}(z)$ remains true.
\begin{enumerate}[resume]
\item \emph{\label{enu:id x}Identify all remaining variables from $\mathrm{Imp}(z)\setminus\{z\}$
with $z$.} 
\end{enumerate}
Now $\mathrm{Imp}(z)\setminus\{z\}$ is empty, so the last step is
easy:
\begin{enumerate}[resume]
\item \emph{\label{enu:id z}Identify all remaining variables other than
$x,y,z$ with $x$.} 
\end{enumerate}
\end{proof}

\section{\label{sec:cpss}Constraint-Projection Separating Sets of Relations}

In this supplemental section we reveal a common property of all CPSS
sets $\mathcal{S}$ of relations and derive a simple algorithm for
\noun{Conn}\textsubscript{C}($\mathcal{S}$).
\begin{defn}
\label{dcp}A set $\mathcal{S}$ of logical relations is \emph{constraint-projection
separating}, if every CNF\textsubscript{C}($\mathcal{S}$)-formula
$\phi$ whose solution graph $G(\phi)$ is disconnected contains a
constraint $C_{i}$ s.t. $G(\phi_{i})$ is disconnected, where $\phi_{i}$
is the projection of $\phi$ to $\mathrm{Var}(C_{i})$.
\end{defn}
\begin{wrapfigure}{r}{0.27\columnwidth}%
\includegraphics[scale=0.65]{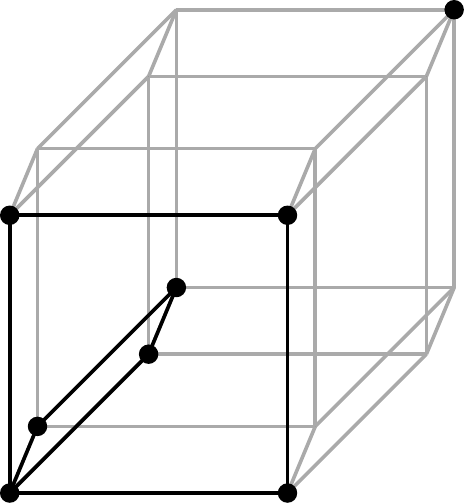}\end{wrapfigure}%
For example, $\{x\vee\overline{y}\}$ is projection-separating (for
the proof see \prettyref{lem:prjs b-1}); so, e.g. for $\left(x\vee\overline{y}\right)\wedge\left(y\vee\overline{z}\right)\wedge\left(z\vee\overline{x}\right)$,
the projections to $\{x,y\}$, $\{y,z\}$ and $\{z,x\}$ are all disconnected.
In contrast, $\{x\vee\overline{y}\vee\overline{z}\}$ is not projection-separating:
E.g., $\left(x\vee\overline{y}\vee\overline{z}\right)\wedge\left(y\vee\overline{z}\vee\overline{w}\right)\wedge\left(z\vee\overline{w}\vee\overline{x}\right)\wedge\left(w\vee\overline{x}\vee\overline{y}\right)$
(see the graph on the right) is disconnected, but the projection to
any three variables is connected.

In \prettyref{lem:prjs b-1} we show that CPSS sets of relations,
as defined in Definition \ref{def:cpss}, are indeed constraint-projection
separating, so that the following algorithm works.
\begin{thm}
\label{thm:alg}Let $\mathcal{S}$ be a CPSS set of relations, and
$\phi$ a \noun{CNF\textsubscript{C}($\mathcal{S}$)}-formula. Then,
the following polynomial-time algorithm decides whether $G(\phi)$
is connected:
\begin{itemize}[label= ]
\item For every constraint $C_{i}$ of $\phi$, obtain the projection $\phi_{i}$
of $\phi$ to the variables $\boldsymbol{x}^{i}$ occurring in $C_{i}$
by checking for every assignment $\boldsymbol{a}$ of $\boldsymbol{x}^{i}$
whether $\phi[\boldsymbol{x}^{i}/\boldsymbol{a}]$ is satisfiable.
Then $G(\phi)$ is connected iff for no $\phi_{i}$, $G(\phi_{i})$
is disconnected.
\end{itemize}
($\phi[\boldsymbol{x}^{i}/\boldsymbol{a}]$ denotes the formula resulting
from $\phi$ by substituting the constants $a_{j}$ for the variables
$x_{j}^{i}$.)\end{thm}
\begin{proof}
Every projection can be computed in polynomial time since \noun{$\mathcal{S}$}
is Schaefer, and connectivity of every $G(\phi_{i})$ can be checked
in constant time. If $G(\phi)$ is disconnected, some $G(\phi_{i})$
is disconnected since $\phi$ is constraint-projection separating
by \prettyref{lem:prjs b-1} below. If some $G(\phi_{i})$ is disconnected,
it is obvious that $G(\phi)$ cannot be connected.\end{proof}
\begin{lem}
\label{lem:prjs i}Let $\mathcal{S}$ be a set of IHSB$-$ (IHSB$+$)
relations and $\phi$ a \noun{CNF\textsubscript{C}($\mathcal{S}$)}-formula.
Then for any two components of $G(\phi)$, there is some constraint
$C_{i}$ of $\phi$ s.t.$\mbox{\,}$their images in the projection
$\phi_{i}$ of $\phi$ to $\mathrm{Var}(C_{i})$ are disconnected
in $G(\phi_{i})$.\end{lem}
\begin{proof}
We prove the IHSB$-$ case, the IHSB$+$ case is analogous. Consider
any two components $A$ and $B$ of $\phi$. Since every IHSB$-$
relation is OR-free, there is a locally minimal solution $\boldsymbol{a}$
in $A$ and a locally minimal solution $\boldsymbol{b}$ in $B$ by
Lemma 4.5 of \citep{gop}. Let $U$ and $V$ be the sets of variables
that are assigned 1 in $\boldsymbol{a}$ and \textbf{$\boldsymbol{b}$},
resp.~At least one of the sets $U'=U\setminus V$ or $V'=V\setminus U$
is not empty, assume it is $U'$. Then for every $x_{1}\in U'$ there
must be a clause $x_{1}\vee\overline{x}_{2}$ with $x_{2}\in U$ since
$\boldsymbol{a}$ is locally minimal, and also $x_{2}$ must be from
$U'$, else $\boldsymbol{b}$ would not be satisfying.

But then for $x_{2}$ there must be also some variable $x_{3}\in U'$
and a clause $x_{2}\vee\overline{x}_{3}$, and we can add the clause
$x_{1}\vee\overline{x}_{3}$ to $\phi$ without changing its value.
Continuing this way, we will find a cycle, i.e.~a clause $x_{i}\vee\overline{x}_{i+1}$
with $x_{i+1}=x_{j}$, $j<i$. But then we already have $x_{j}\vee\overline{x}_{i}$
added, thus $(s_{i},s_{j})\in\{(0,0),(1,1)\}$ for any solution $\boldsymbol{s}$
of $\phi$, and there must be some constraint $C_{i}$ with both $x_{i}$
and $x_{j}$ occurring in it (the $C_{i}$ in which the original $x_{i}\vee\overline{x}_{j}$
appeared), and thus the projections of $A$ and $B$ to $\mathrm{Var}(C_{i})$
are disconnected in $G(\phi_{i})$.\end{proof}
\begin{lem}
\label{lem:prjs b}Let $\mathcal{S}$ be a set of bijunctive relations
and $\phi$ a \noun{CNF\textsubscript{C}($\mathcal{S}$)}-formula.
Then for any two components of $G(\phi)$, there is some constraint
$C_{i}$ of $\phi$ s.t.$\mbox{\,}$their images in the projection
$\phi_{i}$ of $\phi$ to $\mathrm{Var}(C_{i})$ are disconnected
in $G(\phi_{i})$.\end{lem}
\begin{proof}
The proof is similar to the last one. Consider any two components
$A$ and $B$ of $\phi$ and two solutions $\boldsymbol{a}$ in $A$
and $\boldsymbol{b}$ in $B$ that are at minimum Hamming distance.
Let $L$ be the set of literals that are assigned 1 in \textbf{$\boldsymbol{a}$},
but assigned 0 in\textbf{ $\boldsymbol{b}$}. Then for every $l_{1}\in L$
that is assigned 1 in\textbf{ $\boldsymbol{a}$}, there must be a
clause equivalent to $l_{1}\vee\overline{l}_{2}$ in $\phi$ s.t.$\mbox{\,}$$l_{2}$
is also assigned 1 in\textbf{ $\boldsymbol{a}$}, else the variable
corresponding to $l_{1}$ could be flipped in \textbf{$\boldsymbol{a}$},
and the resulting vector would be closer to \textbf{$\boldsymbol{b}$},
contradicting our choice of $\boldsymbol{a}$ and \textbf{$\boldsymbol{b}$}.
Also, $l_{2}$ must be assigned 0 in \textbf{$\boldsymbol{b}$}, i.e.~$l_{2}\in L$,
else $\boldsymbol{b}$ would not be satisfying. 

But then for $l_{2}$ there must be also some literal $l_{3}\in L$
that is assigned 1 in\textbf{ }$\boldsymbol{a}$ and a clause equivalent
to $l_{2}\vee\overline{l}_{3}$ in $\phi$, and we can add the clause
$l_{1}\vee\overline{l}_{3}$ to $\phi$ without changing its value.
Continuing this way, we will find a cycle, i.e.~a clause equivalent
to $l_{n}\vee\overline{l}_{n+1}$ with $l_{n+1}=l_{m}$, $m<n$. But
then we already have $l_{m}\vee\overline{l}_{n}$ added, thus if $x_{i}$
and $x_{j}$ are the variables corresponding to $l_{n}$ resp.~$l_{m}$,
then $(s_{i},s_{j})\in\{(0,1),(1,0)\}$ (if $l_{n}$ and $l_{m}$
were both positive or both negative), or $(s_{i},s_{j})\in\{(0,0),(1,1)\}$
(otherwise), for any solution $\boldsymbol{s}$ of $\phi$. Also,
there must be some constraint $C_{i}$ with both $x_{i}$ and $x_{j}$
occurring in it (the constraint in which the clause equivalent to
$l_{n}\vee\overline{l}_{m}$ appeared), and thus the projections of
$A$ and $B$ to $\mathrm{Var}(C_{i})$ are disconnected in $G(\phi_{i})$.\end{proof}
\begin{lem}
\label{lem:prjs b-1}Every set $\mathcal{S}$ of safely componentwise
bijunctive (safely componentwise IHSB$-$, safely componentwise IHSB$+$,
affine) relations is constraint-projection separating.\end{lem}
\begin{proof}
The affine case follows from the safely componentwise bijunctive case
since every affine relation is safely componentwise bijunctive by
Lemma 4.2 of \citep{gop}.

If the relation corresponding to some $C_{i}$ is disconnected, and
there is more than one component of this relation for which $\phi$
has solutions with the variables of $C_{i}$ assigned values in that
component, the projection of $\phi$ to $\mathrm{Var}(C_{i})$ must
be disconnected in $G(\phi_{i})$.

So assume that for every constraint $C_{i}$, $\phi$ only has solutions
in which the variables of $C_{i}$ are assigned values in one component
$P_{i}$ of the relation corresponding to $C_{i}$. Then we can replace
every $C_{i}$ with $P_{i}$ to obtain an equivalent formula $\phi'$.
Since $\mathcal{S}$ is safely componentwise bijunctive (safely componentwise
IHSB$-$, safely componentwise IHSB$+$), each $P_{i}$ is bijunctive
(IHSB$-$, IHSB$-$), and thus so is $\phi'$, and the statement follows
from Lemmas \ref{lem:prjs i} and \ref{lem:prjs b}. \end{proof}
\begin{rem}
The Lemmas \ref{lem:prjs i} and \ref{lem:prjs b} cannot be generalized
to safely componentwise bijunctive or safely componentwise IHSB$-$
relations: For sets \emph{$\mathcal{S}$} of safely componentwise
bijunctive (safely componentwise IHSB$-$) relations that are not
bijunctive (IHSB$-$), there are CNF\textsubscript{C}($\mathcal{S}$)-formulas
with pairs of components that are not disconnected in the projection
to any constraint.

E.g., for the safely componentwise bijunctive relation $R=\left((x\vee\overline{y})\wedge\overline{z}\right)\vee\left(\overline{x}\wedge y\wedge z\right)$,
the CNF\textsubscript{C}(\emph{$\{R\}$})-formula $F(x,y,z,w)=R(x,y,z)\wedge R(y,x,w)$
has the four pairwise disconnected solutions $a$=0000, $b$=1100,
$c$=0110, and $d$=1001, but $a$ is connected to $b$ in the projection
to $\{x,y,z\}$ as well as in the one to $\{x,y,w\}$.
\end{rem}
Finally, we show that Schaefer sets of relations that are not CPSS
are not constraint-projection separating. \prettyref{lem:prjs b-1}
shows that there are non-Schaefer sets that are constraint-projection
separating. It is open whether there are other such sets not mentioned
in \prettyref{lem:prjs b-1}.
\begin{lem}
If a set of relations $\mathcal{S}$ is Schaefer but not CPSS, there
is a \noun{CNF\textsubscript{C}($\mathcal{S}$)}-formula $\phi$
that is not constraint-projection separating.\end{lem}
\begin{proof}
Since $\mathcal{S}$ is Schaefer but not CPSS, it must contain some
relation that is Horn but not safely componentwise IHSB$-$, or dual
Horn but not safely componentwise IHSB$+$. Assume the first case,
the second one is analogous. Then by \prettyref{lem:exp m}, we can
express $M=\left(x\vee\overline{y}\vee\overline{z}\right)\wedge\left(\overline{x}\vee z\right)$
as a CNF\textsubscript{C}($\mathcal{S}$)-formula. Consider the CNF\textsubscript{C}($\mathcal{S}$)-formula
\[
T(u,v,w,x,y,z)=M(u,v,w)\wedge M(x,y,z)\wedge M(w,w,y)\wedge M(z,z,v)
\]
\[
\equiv\left(\left(u\vee\overline{v}\vee\overline{w}\right)\wedge\left(\overline{u}\vee w\right)\right)\wedge\left(\left(x\vee\overline{y}\vee\overline{z}\right)\wedge\left(\overline{x}\vee z\right)\right)\wedge\left(y\vee\overline{w}\right)\wedge\left(v\vee\overline{z}\right).
\]
Now $G(T)$ is disconnected by \prettyref{cor:horn conn} since $\{u,v,w,x,y,z\}$
is maximal self-implicating, but neither the projection $\exists x\exists y\exists zT\equiv M(u,v,w)$
to the variables of the first constraint in the CNF($\{M\}$)-representation
of $T$, nor the projection $\exists u\exists v\exists x\exists zT\equiv y\vee\overline{w}$
to the variables of the third one is disconnected. The second and
fourth constraints are symmetric to the first and third ones. 

Since in the CNF\textsubscript{C}($\mathcal{S}$)-representation
of $T$ every conjunct $M(r,s,t)$ of $T$ ($r,s,t\in\{u,v,w,x,y,z\}$)
is a CNF\textsubscript{C}($\mathcal{S}$)-formula $\bigwedge_{i}R_{i}(\boldsymbol{\xi}^{i})$
with $R_{i}\in\mathcal{S}$ and $\xi_{j}^{i}\in\{0,1,r,s,t\}$, for
every constraint $C_{i}$ of $T$, the set $\mathrm{Var}(C_{i})$
is a subset of $\{u,v,w\},$ $\{x,y,z\},$ $\{y,w\}$ or $\{v,z\}$,
and thus also for no $C_{i}$ the projection to $\mathrm{Var}(C_{i})$
is disconnected.
\end{proof}

\section{Related and Future Work}

\paragraph*{Other classification schemes and other representations of Boolean
relations}

As Gopalan et al.\ already remarked \citep{gop}, a classification
for CNF($\mathcal{S}$)-formulas without constants seems interesting.
In \citep{diss}, we show that for $st$-connectivity and the diameter,
the same dichotomy as for CNF\textsubscript{C}($\mathcal{S}$)-formulas
also holds in the no-constants case; for connectivity, we identify
fragments where the problem is in P, where it is coNP-complete, and
where it is PSPACE-complete, but a complete classification is still
missing. 

Another variation are partially quantified formulas (with constants),
for which we prove a complete classification in \citep{diss} for
both problems.

Disjunctive normal forms with special connectivity properties were
studied by Ekin et al. already in 1997 for their ``important role
in problems appearing in various areas including in particular discrete
optimization, machine learning, automated reasoning, etc.'' \citep{ekin1999connected}.

A quite different kind of representation for Boolean relations are
$B$-formulas, i.e. arbitrarily nested formulas built from some finite
set $B$ of connectives (where the arity may be greater than two).
Related are $B$-circuits, which are Boolean circuits where the gates
are from a finite set $B$. In \citep{csr}, we investigate both $B$-formulas
and $B$-circuits and obtain a common dichotomy for the diameter and
both connectivity problems: on one side, the diameter is linear, and
both problems are in P, while on the other, the diameter can be exponential,
and the problems are PSPACE-complete

There are yet more kinds of representations of Boolean relations,
such as binary decision diagrams and or Boolean neural networks, and
investigating the connectivity in these settings might be worthwhile
as well.

\paragraph*{Related problems}

Other connectivity-related problems already mentioned by Gopalan et
al.\ are counting the number of components and approximating the
diameter. 

Further, especially with regard to reconfiguration problems, it is
interesting to find the shortest path between two solutions; this
was recently investigated by Mouawad et al. \citep{mouawad2014shortest},
who proved a computational trichotomy for this problem. In this direction,
one could also consider the optimal path according to some other measure.

\paragraph*{Other definitions of connectivity}

Our definition of connectivity is not the only sensible one: One could
regard two solutions connected whenever their Hamming distance is
at most $d$, for any fixed $d\geq1$; this was already considered
related to random satisfiability, see \citep{achlioptas2006solution}.
This generalization seems meaningful as well as challenging.

\paragraph*{Higher domains}

Finally, a most interesting subject are CSPs over larger domains;
in 1993, Feder and Vardi conjectured a dichotomy for the satisfiability
problem over arbitrary finite domains \citep{feder1998computational},
and while the conjecture was proved for domains of size three in 2002
by Bulatov \citep{bulatov2002dichotomy}, it remains open to date
for the general case. Close investigation of the solution space might
lead to valuable insights here.

For $k$-colorability, which is a special case of the general CSP
over a $k$-element set, the connectivity problems and the diameter
were already studied by Bonsma and Cereceda \citep{bonsma2009finding},
and Cereceda, van den Heuvel, and Johnson \citep{cereceda2011finding}.
They showed that for $k=3$ the diameter is at most quadratic in the
number of vertices and the $st$-connectivity problem is in P, while
for $k\geq4$, the diameter can be exponential and $st$-connectivity
is PSPACE-complete in general.

\paragraph{Acknowledgments.}

I am grateful to Heribert Vollmer for pointing me to these interesting
themes.

\bibliographystyle{amsalpha}
\bibliography{jib}

\end{document}